\documentclass[preprint,12pt]{elsarticle}
\usepackage{amsmath}
\usepackage{amsfonts}
\usepackage{amssymb}
\usepackage{amsthm}
\usepackage{graphicx}
\usepackage{verbatim}
\usepackage{subfigure}

\usepackage{dsfont}
\usepackage[numbers]{natbib}

\usepackage{subfigure}

\newtheorem{thm}{Theorem}[section]

\newtheorem{defn}[thm]{Definition}
\newtheorem{rem}[thm]{Remark}
\newtheorem{prop}[thm]{Proposition}

\newcommand{\levy}{{L\'evy }}

\def\1{\mathds{1}}
\def\R{\mathbb{R}}
\def\P{\mathrm{P}}
\def\Q{\mathrm{Q}}
\def\E{\mathrm{E}}
\def\N{\mathbb{N}}
\def\d{\,\mathrm{d}}
\def\dd{\mathrm{d}}

\def\var{\mathrm{var}}
\def\cov{\mathrm{cov}}
\def\th{\theta}
\def\Th{\Theta}
\def\e{\mathrm{e}}
\def\b{\beta}
\def\g{\gamma}
\def\G{\Gamma}
\def\a{\alpha}

\def\l{\lambda}
\def\law{\overset{\textnormal{law}}{=}}
\def\tp{\top}
\def\F{\mathcal{F}}

\def\Borel{\mathcal{B}(\R)}
\def\k{\kappa}
\def\t{\tau}
\def\h{\psi}
\def\D{\Delta}
\def\B{\mathrm{B}}
\def\ab{\boldsymbol{\alpha}}
\def\xib{\boldsymbol{\xi}}

\begin{document}

\bibliographystyle{model1b-num-names}
\begin{frontmatter}

\title{Archimedean Survival Processes}
\author[label1,labelemail]{Edward Hoyle}
\author[label2]{Levent Ali Meng\"ut\"urk}

\address[label1]{Fulcrum Asset Management, 5--7 Chesterfield Gardens, London W1J 5BQ, UK}
\address[label2]{Department of Mathematics, Imperial College London, London SW7 2AZ, UK}
\address[labelemail]{\textnormal{\texttt{ed.hoyle@fulcrumasset.com}}}

\begin{abstract}
Archimedean copulas are popular in the world of multivariate modelling as a result of their breadth, tractability, and flexibility.
\citeauthor{MN2009} (2009) showed that the class of Archimedean copulas coincides with the class of positive multivariate $\ell_1$-norm symmetric distributions.
Building upon their results, we introduce a class of multivariate Markov processes that we call `Archimedean survival processes' (ASPs).
An ASP is defined over a finite time interval, is equivalent in law to a vector of independent gamma processes, and its terminal value has an Archimedean survival copula.
There exists a bijection from the class of ASPs to the class of Archimedean copulas.
We provide various characterisations of ASPs, and a generalisation.
\end{abstract}


\end{frontmatter}

\section{Introduction} \label{sec:intro}
The use of copulas has become commonplace for dependence modelling in finance, insurance, and risk management (see, for example, \citet{CLV2004}, \citet{FV1998}, and \citet{MFE2005}).
From a modelling perspective, one of the attractive features of copulas is that they allow the fitting of one-dimensional marginal distributions to be performed separately from the fitting of cross-sectional dependence.

The Archimedean copulas---a subclass of copulas---have received particular attention in the literature for both their tractability and practical convenience (see, for example, \citeauthor{GM1986a} \citep{GM1986b,GM1986a} and \citet[chap.~4]{N1999}).
We introduce a family of multivariate stochastic processes that we call \emph{Archimedean survival processes} (ASPs).
ASPs are constructed in such a way that they are naturally linked to Archimedean copulas.
An ASP is defined over a finite time horizon, and its terminal value has a multivariate $\ell_1$-norm symmetric distribution. 
This implies that the terminal value of an ASP has an Archimedean survival copula.
Indeed, there is a bijection from the class of Archimedean copulas to the class of ASPs.

\citet{RN2} suggested using a randomly-scaled gamma bridge (also called a Dirichlet process) for modelling the cumulative payments made on insurance claims (see also \citet{BHM3}).
Such a process $\{\xi_{tT}\}_{0\leq t\leq T}$ can be constructed as $\xi_{tT}=R  \g_{tT}$,
where $R$ is a positive random variable independent of a gamma bridge $\{\g_{tT}\}$ satisfying $\g_{0T}=0$ and $\g_{TT}=1$, for some $T\in(0,\infty)$.
This is an increasing process and so lends itself to the modelling of cumulative gains or losses; 
in this case the random variable $R$ represents the total, final gain.
We can interpret $R$ as a signal and the gamma bridge $\{\g_{tT}\}$ as independent multiplicative noise.
The process $\{\xi_{tT}\}$ can be considered to be a gamma process conditioned so that $\xi_{TT}$ has the law of $R$,
and so belongs to the class of \levy random bridges (see \citet{HHM1}).
As such, we call the process $\{\xi_{tT}\}$ a `gamma random bridge' (GRB).

ASPs are an $n$-dimensional extension of gamma random bridges.
Each one-dimensional marginal process $\{\xi^{(i)}_t\}$ of an ASP $\{(\xi_t^{(1)},\ldots,\xi_t^{(n)})^\tp\}_{0\leq t \leq T}$ is a GRB.
We shall construct each $\{\xi_t^{(i)}\}$ by splitting a `master' GRB into $n$ non-overlapping subprocesses. 
This method of splitting a \levy random bridge into subprocesses (which are themselves \levy random bridges) was used by \citet{HHM2} to develop a bivariate insurance reserving model based on random bridges of the stable-1/2 subordinator. 
A remarkable feature of the proposed construction is that the terminal vector $(\xi_T^{(1)},\ldots,\xi_T^{(n)})^\tp$ has a multivariate $\ell_1$-norm symmetric distribution, and hence an Archimedean survival copula. 

We shall also construct Liouville processes by splitting a GRB into $n$ pieces.
By allowing more flexibility in the splitting mechanism and by employing some deterministic time changes, a broader range of behaviour can be achieved by Liouville processes than ASPs.
For example, the one-dimensional marginal processes of a Liouville process are in general not identical in law.

A direct application of ASPs and Liouville processes is to the modelling of multivariate cumulative gain (or loss) processes.
Consider, for example, an insurance company that underwrites several lines of motor business (such as personal motor, fleet motor or private-hire vehicles) for a given accident year.
A substantial payment made on one line of business is unlikely to coincide with a substantial payment made on another line of business 
(e.g.~a large payment is unlikely to be made on a personal motor claim at the same time as a large payment is made on a fleet motor claim).
However, the total sums of claims arising from the lines of business will depend on certain common factors such as prolonged periods of adverse weather or the quality of the underwriting process at the company.
Such common factors will produce dependence across the lines.
An ASP or a Liouville process might be a suitable model for the cumulative paid-claims processes of the lines of motor business.
The one-dimensional marginal processes of a Liouville process are increasing and do not exhibit simultaneous large jumps, but they can display strong correlation.

ASPs can be used to interpolate a dependence structure when using Archimedean copulas in discrete-time models.
Consider a risk model where the marginal distributions of the returns on $n$ assets are fitted for the future dates $t_1<\cdots<t_n<T<\infty$.
An Archimedean copula $C$ is used to model the dependence of the returns to time $T$.
At this stage we have a model for the joint distribution of returns to time $T$, but we have only the one-dimensional marginal distributions at the intertemporal times $t_1,\ldots,t_n$.
The problem then is to choose copulas to complete the joint distributions of the returns to the times $t_1,\ldots,t_n$ in a way that is consistent with the time-$T$ joint distribution. 
For each time $t_i$, this can be achieved by using the time-$t_i$ survival copula implied by the ASP with survival copula $C$ at terminal time $T$.

This paper is organized as follows: In Section \ref{sec:prelim}, we review multivariate $\ell_1$-norm symmetric distributions, multivariate Liouville distributions, Archimedean copulas and gamma random bridges. 
In Section \ref{sec:ASP}, we define ASPs and provide various characterisations of their law. 
We detail how to construct a multivariate process such that each one-dimensional marginal is uniformly distributed.
An application is then given where an ASP is used to solve an Archimedean copula interpolation problem.
In Section \ref{sec:LP}, we generalise ASPs to Liouville processes.
We apply Liouville processes to the intraday forecasting of realized variance.
In Section \ref{sec:conclude}, we state our conclusions.

\section{Preliminaries} \label{sec:prelim}

We fix a probability space $(\Omega,\P,\F)$ and assume that all processes under consideration are c\`adl\`ag, and all filtrations are right-continuous. 
We let $f^{-1}$ denote the generalised inverse of a monotonic function $f$.
Thus, if $f$ is decreasing then $f^{-1}(y)=\inf\{x: f(x)\leq y\}$.
We denote the $\ell_1$ norm of a vector $\mathbf{x}\in\R^n$ by $\|\mathbf{x}\|$, i.e.~$\|\mathbf{x}\|=\sum_{i=1}^n |x_i|$.

We present some definitions and results from the theory of multivariate distributions and refer the reader to \citet{FKN1990} for further details.

Let $\mathbf{G}$ be a vector of independent random variables such that $G_i$ is a gamma random variable with shape parameter $\a_i>0$ and unit scale parameter. 
Then the random vector $\mathbf{D}=\mathbf{G}/\|\mathbf{G}\|$ has a \emph{Dirichlet distribution} with \emph{parameter vector} $\ab=(\a_1,\ldots,\a_n)^\tp$.
In two dimensions, a Dirichlet random variable can be written as $(B,1-B)^\tp$, where $B$ is a beta random variable. 
If all the elements of the parameter vector $\ab$ are identical, then $\mathbf{D}$ is said to have a \emph{symmetric} Dirichlet distribution. 

A random variable $\mathbf{X}$ taking values in $\R^n$ has a \emph{multivariate Liouville distribution} if $\mathbf{X}\law R \mathbf{D}$,
for $R\geq 0$ a random variable, and $\mathbf{D}$ a Dirichlet random variable, independent of $R$, with parameter vector $\ab$. 
We call the law of $R$ the \emph{generating law} and $\ab$ the \emph{parameter vector} of the distribution.
In the case where $R$ is positive and has a density $p$, the density of $\mathbf{X}$ exists and can be written as
	\begin{equation}
		\label{eq:MLD}
		\mathbf{x}\mapsto \G(\|\ab\|)\frac{p\left(\|\mathbf{x}\|\right)}{\left(\|\mathbf{x}\|\right)^{\|\ab\|-1}} \prod_{i=1}^n \frac{x_i^{\a_i-1}}{\G(\a_i)},
	\end{equation}
for $\mathbf{x}\in\R_+^n$, where $\G(x)$ is the gamma function \citep[6.1]{AS1964}.
In the case $\ab=(1,\ldots,1)^{\tp}$, $\mathbf{X}$ has a \emph{multivariate $\ell_1$-norm symmetric distribution}.
A multivariate $\ell_1$-norm symmetric distribution is characterised by its generating law.

\citet{MN2009} give an account of how Archimedean copulas coincide with survival copulas of $\ell_{1}$-norm symmetric distributions which have no point-mass at the origin.
Then in \citep{MN2010}, \citeauthor{MN2010} generalise Archimedean copulas to so-called Liouville copulas, which are defined as the survival copulas of multivariate Liouville distributions.

A copula is a distribution function on the unit hypercube with the added property that each one-dimensional marginal distribution is uniform.
Archimedean copulas are copulas that take a particular functional form. 
The following definition given in \citep{MN2009} is convenient for the present work:
	A decreasing and continuous function $\h:[0,\infty)\rightarrow[0,1]$ which satisfies the conditions $\h(0)=1$ and $\lim_{x\rightarrow\infty}\h(x)=0$, 
	and is strictly decreasing on $[0,\inf\{x:\h(x)=0\})$ is called an \emph{Archimedean generator}.
	An $n$-dimensional copula $C$ is called an \emph{Archimedean copula} if it permits the representation
	\[ C(\mathbf{u})=\h(\h^{-1}(u_1)+\cdots+\h^{-1}(u_n)),  \qquad \mathbf{u}\in[0,1]^n,\]
	for some Archimedean generator $\h$ with inverse $\h^{-1}:[0,1]\rightarrow[0,\infty)$, where we set $\h(\infty)=0$ and $\h^{-1}(0)=\inf\{u:\h(u)=0\}$.

If $\mathbf{X}$ has an $n$-variate $\ell_1$-norm symmetric distribution with generating law $\nu$ and $\P(\mathbf{X}=\mathbf{0})=0$,
then $\mathbf{X}$ has an Archimedean survival copula with generator
\begin{equation}
	\label{eq:bijection1}
	\h(x)=\P(X_i>x)=\int_x^\infty (1-x/r)^{n-1} \nu(\dd r).
\end{equation}
\citet{MN2009} showed that the converse is also true: 
If $\mathbf{U}$ has an $n$-dimensional Archimedean copula $C$ with generator $\psi$,
then $(\psi^{-1}(U_1),\ldots,\psi^{-1}(U_n))^{\tp}$ has a multivariate $\ell_1$-norm distribution with survival copula $C$ and generating law $\nu$ given by
\begin{equation} \label{eq:bijection2}
	\nu([0,x])=1-\sum_{k=0}^{n-2}\frac{(-1)^kx^k\psi^{(k)}(x)}{k!}-\frac{(-1)^{n-1}x^{n-1} \psi^{(n-1)}_+(x)}{(n-1)!}, \quad x\geq 0,
\end{equation}
where $\psi^{(k)}$ is the $k$th derivative of $\psi$, and $\psi^{(n-1)}_+$ is the right-hand sided derivative of order $n-1$.

A gamma random bridge is an increasing stochastic process, and both the gamma process and gamma bridge are special cases.
A gamma process is a subordinator (an increasing \levy process) with gamma distributed increments (see, for example, \citet{Sato1999}). 
The law of a gamma process is uniquely determined by its mean and variance at time 1, which are both positive. 
Let $\{\g_t\}$ be a gamma process with mean and variance $m>0$ at time 1; then $\E(\g_t)=mt$, and $\var(\g_t)=mt$.
The density of $\g_t$ is $f_t(x;m)=  \1_{\{x>0\}}x^{mt-1}\e^{-x}/\G(mt)$.

A gamma bridge is a gamma process conditioned to have a fixed value at a fixed future time. 
A gamma bridge is a \levy bridge, and hence a Markov process. 
Let $\{\g_{tT}\}_{0\leq t\leq T}$ be a gamma bridge identical in law to the gamma process $\{\g_t\}$ pinned to the value 1 at time $T$. 
The transition law of $\{\g_{tT}\}$ is given by
\begin{equation}
	\P\left(\g_{tT}\in \dd y \left|\, \g_{sT}=x\right.\right)=\1_{\{x<y< 1\}}\frac{\left(\frac{y-x}{1-x}\right)^{m(t-s)-1} \left(\frac{1-y}{1-x}\right)^{m(T-t)-1}}
													{(1-x)\mathrm{B}(m(t-s),m(T-t))} \d y, \label{eq:gambrtrans}
\end{equation}
for $0\leq s <t\leq T$ and $x\geq 0$.
Here $\mathrm{B}(\a,\b)$ is the beta function \citep[6.2]{AS1964}.
We say that $m$ is the \emph{activity parameter} of $\{\g_{tT}\}$.
If the gamma bridge $\{\g_{tT}\}$ has reached the value $x$ at time $s$, then it must yet travel a distance $1-x$ over the time period $(s,T]$.
Equation (\ref{eq:gambrtrans}) shows that the proportion of this distance that the gamma bridge will cover over $(s,t]$ is a random variable with a beta distribution.

It is a property of gamma processes that the renormalised process $\{\g_t/\g_T\}_{0\leq t \leq T}$ is independent of $\g_T$.
This leads to the remarkable identity $\left\{ \g_t/\g_T \right\}\law \{\g_{tT}\}$, which we refer to as the \emph{ratio property} of the gamma bridge.
It follows that the joint distribution of increments of a gamma bridge is Dirichlet.

\begin{defn}
	\label{def:GRB}
	The process $\{\G_{t}\}_{0\leq t \leq T}$ is a \emph{gamma random bridge (GRB)} if
	\begin{equation} \label{eq:GRB}
		\{\G_{t}\}\law \{R \g_{tT}\},
	\end{equation}
	for $R>0$ a random variable, and $\{\g_{tT}\}$ a gamma bridge independent of $R$.
	We say that $\{\G_{t}\}$ has \emph{generating law $\nu$} and \emph{activity parameter $m$}, where of $\nu$ is the law of $R$ and $m$ is the activity parameter of $\{\g_{tT}\}$.
\end{defn}

Suppose that $\{\G_t\}$ is a GRB satisfying \textnormal{(\ref{eq:GRB})}.
If $\P(R=z)=1$ for some $z>0$, then $\{\G_t\}$ is a gamma bridge.
If $R$ is gamma random variable with shape parameter $mT$ and scale parameter $\k$, 
then $\{\G_t\}$ is a gamma process such that $\E(\G_t)=m\k t$ and $\var(\G_t)=m\k^2t$, for $t\in[0,T]$.

GRBs fall within the class of \levy random bridges described in \citep{HHM1}. 
The process $\{\G_{t}\}$ is a Markov process with stationary increments, and is identical in law to a gamma process defined over $[0,T]$ conditioned to have the law of $R$ at time $T$.
The bridges of a GRB are gamma bridges.
Since increments of a gamma bridge have a Dirichlet distribution, it follows that the increments of a GRB have a multivariate Liouville distribution.

Define the subprocesses $\{\xi^{(i)}_{t}\}_{0\leq t \leq T_i}$, $i=1,\ldots,n$, by
\begin{align*}
	\xi^{(i)}_t&=\G_{s_i+t}-\G_{s_i}, && \text{for $t\in[0,T_i]$},
	\\\xi^{(i)}_t&=\xi^{(i)}_{T_i}, && \text{for $t>T_i$},
\end{align*}
where the intervals $[s_i,s_i+T_i]$, $i=1,\ldots,n$, are non-overlapping except possibly at the endpoints. 
It follows from \citep[Corollary 3.12]{HHM1} that each $\{\xi^{(i)}\}$ is a GRB with generating law
\[ \nu^{(i)}(\dd x)=\frac{x^{mT_i-1}}{\B(mT_i,m(T-T_i))} \int_{z=x}^{\infty} z^{1-mT}(z-x)^{m(T-T_i)-1} \nu(\dd z) \d x. \]
Furthermore, we can construct an $n$-dimensional Markov process $\{\xib_t\}$  by setting $\xib_t=(\xi^{(1)}_t,\ldots,\xi^{(n)}_t)^\tp$.

\section{Archimedean survival process} \label{sec:ASP}
We construct an ASP by splitting a gamma random bridge into $n$ non-overlapping subprocesses. 
We start with a `master' GRB $\{\G_t\}_{0\leq t \leq n}$ with activity parameter $m=1$ and generating law $\nu$, where $n\in\N_+$, $n\geq 2$. 
In this section, we write $f_t(x)$ for the gamma density with shape parameter unity and scale parameter unity.
That is $f_t(x)=f_t(x;1)=x^{t-1} \e^{-x}/\G(t)$.

\begin{defn}
	The process $\{\xib_{t}\}_{0\leq t \leq 1}$ is an \emph{$n$-dimensional Archimedean survival process} if
	\begin{equation*}
		\{\xib_t\}_{0\leq t \leq 1}\law\left\{\left[ \begin{aligned}
					& \G_{t}-\G_{0}
					\\ & \vdots
					\\ & \G_{(i-1)+t}-\G_{i-1}
					\\ & \vdots
					\\ & \G_{(n-1)+t}-\G_{n-1}
			\end{aligned} \right]\right\}_{0\leq t \leq 1},
	\end{equation*}
	where $\{\G_t\}_{0\leq t \leq n}$ is a gamma random bridge with activity parameter $m=1$. 
	We say that the generating law of $\{\G_t\}$ is the \emph{generating law} of $\{\xib_{t}\}$.
\end{defn}

\begin{figure}[ht]
	\begin{center}
		\subfigure[GRB]{\includegraphics[scale=0.2]{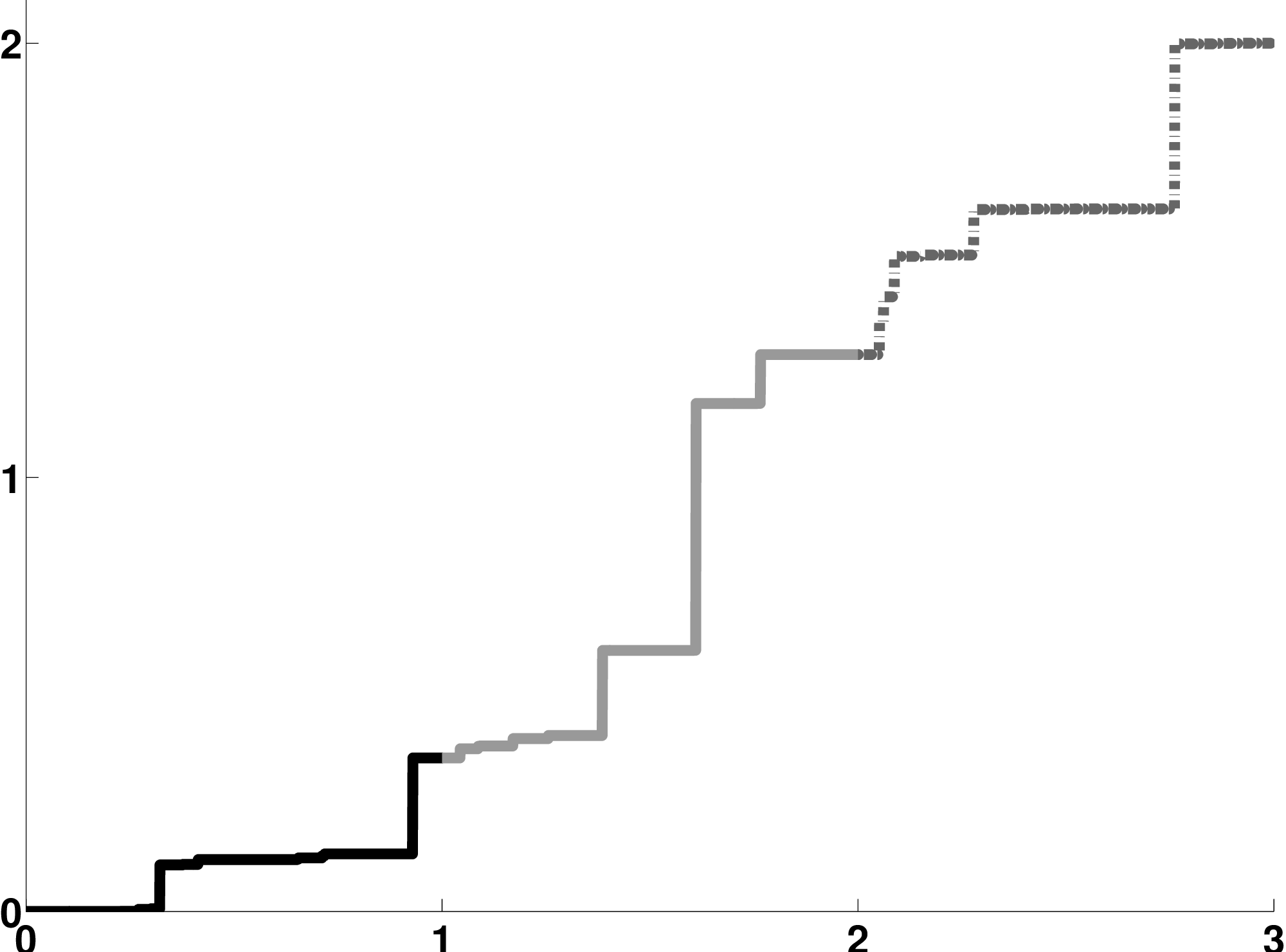}}
		\subfigure[ASP]{\includegraphics[scale=0.2]{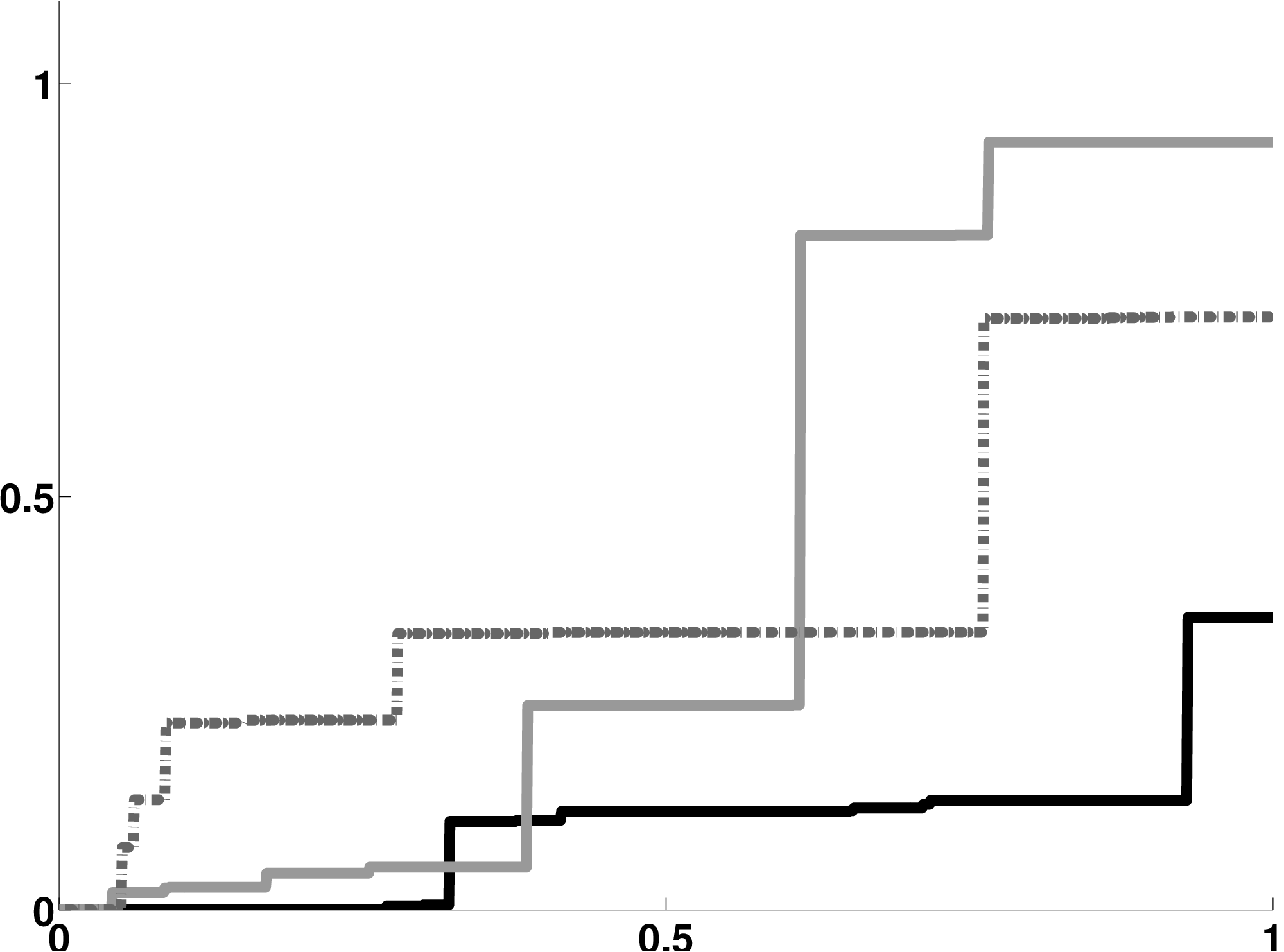}}
	\end{center}
	\caption{%
			A graphical representation of the construction of an ASP.
			A GRB is split into three subprocesses, each spanning a time interval of unit length.
			The subprocesses are spatially and temporally transformed so that they start at value zero at time zero, and terminate at unit time.
			}
	\label{fig:surprise}
\end{figure}

Note that from Definition \ref{def:GRB} $\P(\G_n=0)=0$, and so $\P(\xib_1=\mathbf{0})=0$.
Each one-dimensional marginal process of an ASP is a subprocess of a GRB, and hence a GRB. 
Thus ASPs are a multivariate generalisation of GRBs.
We defined ASPs over the time interval $[0,1]$;
it is straightforward to restate the definition to cover an arbitrary closed interval.

\begin{prop}
	The terminal value of an ASP has an Archimedean survival copula.
\end{prop}
\begin{proof}
	Let $\{\xib_{t}\}$ be an $n$-dimensional ASP with generating law $\nu$. Then we have
	\begin{align*}
		\P(\xib_1\in\dd \mathbf{x})&=\P\left(\G_1-\G_0\in\dd x_1,\ldots,\G_n-\G_{n-1}\in\dd x_n\right)
		\\ &=\P\left(R\frac{\g_1-\g_0}{\g_n}\in\dd x_1,\ldots,R\frac{\g_n-\g_{n-1}}{\g_n}\in\dd x_n\right),
	\end{align*}
	for $\mathbf{x}\in\R^n$, where $R$ is a random variable with law $\nu$ and $\{\g_t\}$ is a gamma process, independent of $R$, such that $\g_t$ has density $f_t(x)=x^{t-1} \e^{-x}/\G(t)$.
	Each increment $\g_i-\g_{i-1}$ has an exponential distribution (with unit rate). 
	Thus $\P(\xib_1\in B)=\P(R \,\mathbf{E}/\|\mathbf{E}\| \in B)$,
	for $\mathbf{E}$ an $n$-vector of independent, identically-distributed, exponential random variables.
	Hence $\xib_1$ has a multivariate $\ell_1$-norm symmetric distribution. Therefore, it has an Archimedean survival copula. 
\end{proof}

\begin{rem}
	Let $g_{i}:\R_{+}\rightarrow\R$ be strictly decreasing for $i=1,\ldots,n$, and let $\{\xib_t\}$ be an ASP. 
	Then the vector-valued process $\{(g_{1}(\xi^{(1)}_{t}),\ldots,g_{n}(\xi^{(n)}_{t}))^\tp\}_{0\leq t \leq 1}$
	has an Archimedean copula at time $t=1$.
\end{rem}

\subsection{Characterisations} \label{subsec:charac}
In this subsection we shall characterize ASPs first through their finite-dimensional distributions, and then through their transition probabilities.

We shall show that the joint distribution of increments of an ASP are multivariate Liouville.
To this end, we first show that the joint distribution of increments of a GRB are multivariate Liouville.
The finite-dimensional distributions of the master process $\{\G_t\}$ are given by
\begin{equation*}
	\P(\G_{t_1} \in \dd x_1, \ldots, \G_{t_k} \in \dd x_k, \G_n\in \dd z)
	 = \P(\G_{t_1} \in \dd x_1, \ldots, \G_{t_k} \in \dd x_k \,|\, \G_n=z) \, \nu(\dd z),
\end{equation*}
where $x_0=0$, for all $k\in\N_+$, all partitions $0=t_0<t_1<\cdots<t_k<n$, all $z\in\R_+$, and all $(x_1,\ldots,x_k)^\tp=\mathbf{x} \in \R_+^k$.
It was mentioned earlier that the bridges of a GRB are gamma bridges.
Hence, for $\{\g_t\}$ a gamma process such that $\E(\g_1)=1$ and $\var(\g_1)=1$, we have
\begin{equation*}
	\P[\G_{t_1} \in \dd x_1, \ldots, \G_{t_k} \in \dd x_k, \G_n\in \dd z]
	= \P[\g_{t_1} \in \dd x_1, \ldots, \g_{t_k} \in \dd x_k \,|\, \g_n=z] \, \nu(\dd z).
\end{equation*}
Using the ratio property of the gamma bridge and (\ref{eq:GRB}), we have
\begin{equation*}
	(\G_{t_1}-\G_{t_0},\ldots,\G_{t_k}-\G_{t_{k-1}},\G_n-\G_{t_k}) \law \frac{R}{\g_n}(\g_{t_1}-\g_{t_0},\ldots,\g_{t_k}-\g_{t_{k-1}},\g_n-\g_{t_k}).
\end{equation*}
Hence $(\G_{t_1}-\G_{t_0},\ldots,\G_{t_k}-\G_{t_{k-1}},\G_n-\G_{t_k})^{\tp}$
has a multivariate Liouville distribution with generating law $\nu$ and parameter vector $(t_1-t_0,\ldots,t_k-t_{k-1},n-t_k)^\tp$.
 
We can use these results to characterise the law of the ASP $\{\xib_t\}$ through the joint distribution of its increments.
Fix $k_i\geq 1$ and the partitions $0=t_0^i<\cdots<t_{k_i}^i=1$, for $i=1,\ldots,n$. 
Then define the non-overlapping increments $\{\D_{ij}\}$ by $\D_{ij}=\xi^{(i)}_{t^i_j}-\xi^{(i)}_{t^i_{j-1}}$, for $j=1,\ldots,k_i$ and $i=1,\ldots,n$.
The distribution of the $k_1\cdot \cdots \cdot k_n$-element vector $\boldsymbol{\D}=(\D_{11},\ldots,\D_{1k_1},\ldots,\D_{n1},\ldots,\D_{nk_n})^\tp$
characterises the finite-dimensional distributions of the ASP $\{\xib_t\}$. 
Thus it follows from the Kolmogorov extension theorem that the distribution of $\boldsymbol{\D}$ characterises the law of $\{\xib_t\}$. 
Note that $\boldsymbol{\D}$ contains non-overlapping increments of the master GRB $\{\G_t\}$ such that $\|\boldsymbol{\D}\|=\G_n$.
Hence $\boldsymbol{\D}$ has a multivariate Liouville distribution with parameter vector
$\boldsymbol{\a}=(t^1_1-t^1_0,\ldots,t^1_{k_1}-t^1_{k_1-1},\ldots,t^n_1-t^n_0,\ldots,t^n_{k_n}-t^n_{k_n-1})^\tp$,
and the generating law $\nu$.

\bigskip

We denote the filtration generated by $\{\xib_t\}_{0\leq t \leq 1}$ by $\{\F_t\}$.
Then $\{\xib_t\}$ is a Markov process with respect to $\{\F_t\}$.

For a set $B\subset\R$ and a constant $x\in\R$, we write $B+x$ for the shifted set $B+x=\{y\in\R: y-x\in B\}$.
We define the process $\{R_t\}_{0\leq t \leq 1}$ by setting
\begin{equation*}
	R_t=\sum_{i=1}^n \xi^{(i)}_t=\|\xib_t\|.
\end{equation*}
Note that the terminal value of $\{R_t\}$ is the terminal value of the master process $\{\G_t\}$, i.e.~$R_1=\G_n$.
We define a family of unnormalised measures, indexed by $t\in[0,1)$ and $x\in\R_+$, as follows: $\th_0(B;x)=\nu(B)$ and
\begin{align*}
	\th_t(B;x)=\int_B \frac{f_{n(1-t)}(z-x)}{f_n(z)} \, \nu(\dd z),
\end{align*}
for $B\in\Borel$.
We also write $\Th_t(x)=\th_t([0,\infty);x)$.

\begin{prop} \label{prop:transition}
	The ASP $\{\xib_t\}$ is a Markov process with the transition law given by
	\begin{multline}
		\label{eq:ASP1}
		\P\left(\left. \xi_1^{(1)}\in\dd z_1,\ldots, \xi_1^{(n-1)}\in\dd z_{n-1},\xi_1^{(n)}\in B  \,\right| \xib_s=\mathbf{x} \right)=
		\\ \frac{\th_{\t(s)}(B+\sum_{i=1}^{n-1}z_i;x_n+\sum_{i=1}^{n-1}z_i)}{\Th_s(\|\mathbf{x}\|)}
									\prod_{i=1}^{n-1}\frac{(z_i-x_i)^{-s}\e^{-(z_i-x_i)}}{\G(1-s)}\d z_i,
	\end{multline}
	and
	\begin{equation}
		\label{eq:ASP2}
		\P\left( \xib_t\in \d\mathbf{y}  \,|\, \xib_s=\mathbf{x} \right)=
					\frac{\Th_{t}(\|\mathbf{y}\|)}{\Th_s(\|\mathbf{x}\|)}
					\prod_{i=1}^{n}\frac{(y_i-x_i)^{(t-s)-1}\e^{-(y_i-x_i)}}{\G(t-s)}\d y_i, 
	\end{equation}
	where $\t(t)=1-(1-t)/n$, $0\leq s<t<1$, and $B\in\Borel$.
\end{prop}
\begin{proof}
	We begin by verifying (\ref{eq:ASP1}).
	From the Bayes theorem we have
	\begin{multline}
		\label{eq:A}
		\P\left(\left. \xi_1^{(1)}\in\dd z_1,\ldots, \xi_1^{(n-1)}\in\dd z_{n-1},\xi_1^{(n)}\in B  \,\right| \xib_s=\mathbf{x} \right)=
		\\ \frac{\P\left(\xi_1^{(1)}\in\dd z_1,\ldots, \xi_1^{(n-1)}\in\dd z_{n-1},\|\xib_1\|\in B+\sum_{i=1}^{n-1} z_i, \xib_s\in\d\mathbf{x} \right)}
						{\P\left(\xib_s\in\d\mathbf{x} \right)}.		
	\end{multline}
	From \citep[Section 3.2]{HHM1} we have
	\begin{equation}
		\P(\G_{t_1} \in \dd x_1, \ldots, \G_{t_k} \in \dd x_k, \G_n\in \dd z)
		= \prod_{i=1}^k\{f_{t_i-t_{i-1}}(x_i-x_{i-1}) \d x_i\} \th_{t_k/n}(\dd z;x_k). \label{eq:useful}
	\end{equation}
	The law of $R_1=\|\xib_1\|$ is $\nu$; hence using (\ref{eq:useful}) the numerator of (\ref{eq:A}) is
	\begin{multline}
		\label{eq:num}
		\int_{u\in B+\sum_{i=1}^{n-1} z_i}
			\P\left(\left.\xi_1^{(1)}\in\dd z_1,\ldots, \xi_1^{(n-1)}\in\dd z_{n-1}, \xib_s\in\d\mathbf{x} \,\right| R_1=u \right) \nu(\dd u)=
		\\ \prod_{i=1}^{n}\{f_{s}(x_i)\d x_i\} \prod_{i=1}^{n-1}\{f_{1-s}(z_i-x_i)\d z_i\} 
		\int_{u \in B+\sum_{i=1}^{n-1} z_i} \frac{f_{1-s}(u-x_n-\sum_{i=1}^{n-1}z_i)}{f_n(u)} \,\nu(\dd u),
	\end{multline}
	and the denominator is
	\begin{align}
		\P\left(\xib_s\in\d\mathbf{x}\right)&=\P\left(\G_s\in \dd x_1, \G_{1+s}-\G_{1}\in \dd x_2,\ldots, \G_{n-1+s}-\G_{n-1}\in \dd x_n\right) \nonumber
		 \\ &=\P\left(\G_s\in \dd x_1,\G_{2s}-\G_{s}\in \dd x_2, \ldots, \G_{ns}-\G_{(n-1)s}\in \dd x_n\right) \label{eq:cyclic}
		 \\ &=\prod_{i=1}^{n}\{f_{s}(x_i)\d x_i\} \int_{z=\|\mathbf{x}\|}^\infty  \th_{s}(\dd z;\|\mathbf{x}\|). \label{eq:denom_new}
	\end{align}
	In (\ref{eq:num}) we have used the fact that, given $\|\xib_1\|=R_1$, $\{\xib_t\}$ is a vector of subprocesses of a gamma bridge.
	Equation (\ref{eq:cyclic}) follows from the stationary increments property of GRBs and (\ref{eq:denom_new}) follows from (\ref{eq:useful}).
	Dividing (\ref{eq:num}) by (\ref{eq:denom_new}) yields the claim.

	We shall now verify (\ref{eq:ASP2}) following similar steps.
	We have
	\begin{equation}
		\label{eq:B}
		\P(\xib_t\in\dd \mathbf{y} \,|\, \xib_s= \mathbf{x})
			=\frac{\P(\xib_t\in\dd \mathbf{y} , \xib_s\in\dd \mathbf{x} )}{\P(\xib_s\in\dd \mathbf{x})}.
	\end{equation}
	The numerator of (\ref{eq:B}) is
	\begin{multline}
		\label{eq:num2}
		\int_{z=0}^{\infty} \P\left(\xib_t\in\dd \mathbf{y} , \xib_s\in\dd \mathbf{x}  \,|\, R_1=z \right) \, \nu(\dd z)=
		\\ \prod_{i=1}^n\{f_{s}(x_i) \d x_i\} \prod_{i=1}^n\{f_{t-s}(y_i-x_i) \d y_i\} \int_{z=0}^{\infty} \frac{f_{n(1-t)}(z-\|\mathbf{y}\|)}{f_n(z)} \, \nu(\dd z),
	\end{multline}
	and the denominator is given in (\ref{eq:denom_new}).
	Dividing (\ref{eq:num2}) by (\ref{eq:denom_new}) yields	the result.
\end{proof}

\begin{rem} \label{rem:ASPtransitiondensity}
	When the generating law $\nu$ admits a density $p$, \textnormal{(\ref{eq:A})} is equivalent to 
	\begin{equation}
		\label{eq:nu_den_td}
		\P\left( \xib_1\in \d\mathbf{z} \,|\, \xib_s=\mathbf{x} \right)=
				\frac{\G(n) \e^{\|\mathbf{x}\|} p(\|\mathbf{z}\|)}{\Th_s(\|\mathbf{x}\|)\|\mathbf{z}\|^{n-1}} \prod_{i=1}^n \frac{(z_{i}-x_i)^{-s}}{\G(1-s)}\d z_i.
	\end{equation}
\end{rem}

\subsubsection{Increments of ASPs}
We shall show that increments of an ASP have $n$-dimensional Liouville distributions.
Indeed, at time $s\in[0,1)$, the increment $\xib_t-\xib_s$, $t\in(s,1]$, has a multivariate Liouville distribution with a generating law that can be expressed in terms of the $\xib_s$-conditional law of the norm variable $R_t=\|\xib_t\|$.
Before we show this, we first examine the law of the process $\{R_t\}$.

We define the measure $\nu_{st}$, $0\leq s<t\leq 1$, by
\begin{equation*}
	\nu_{st}(B)=\P(R_t \in B \,|\, \xib_s), \qquad \text{for $B\in\Borel$}.
\end{equation*}

\begin{prop} \label{prop:R_GRB}
	The process $\{R_t\}_{0\leq t\leq T}$ is a GRB with generating law $\nu$ and activity parameter $n$.
	That is,
	\begin{align}
					\nu_{st}(\dd r)&= \frac{\Th_{t}(r)}{\Th_{s}(\|\xib_s\|)}
																	\frac{(r-\|\xib_s\|)^{n(t-s)-1}\exp\{-(r-\|\xib_s\|)\}}{\G(n(t-s))} \dd r, \label{eq:nu_st}
			\\\intertext{and}
			\nu_{s1}(\dd r)&= \frac{\th_{s}(\dd r;\|\xib_s\|)}{\Th_{s}(\|\xib_s\|)}, \label{eq:nu_s1}
	\end{align}
	for $0<s<t<1$.
\end{prop}

After simplification, (\ref{eq:nu_st}) and (\ref{eq:nu_s1}) are consistent with the transition probabilities given for a GRB in \citep[Section 3.4]{HHM1}.

\begin{proof}
	Since $\{\xib_t\}$ is a Markov process with respect to $\{\F_t\}$, $\{R_t\}$ is a Markov process with respect to $\{\F_t\}$.
	Thus to prove the proposition we need only verify that the transition probabilities of $\{R_t\}$ match those given in (\ref{eq:nu_st}) and (\ref{eq:nu_s1}).
	We achieve this using Bayes' theorem.
	
	The $\xib_s$-conditional law of $R_1$ follows from (\ref{eq:denom_new}):
	\[
		\P(R_1\in \dd r\,|\, \xib_s=\mathbf{x}) = \frac{\P(\xib_s\in\dd \mathbf{x} , R_1=r)}
																			{\P(\xib_s\in\dd \mathbf{x})} 
							= \frac{\th_{s}(\dd r;\|\mathbf{x}\|)}{\Th_{s}(\|\mathbf{x}\|)}.
	\]	
	The $\xib_s$-conditional law of $R_t$ for $t\in(s,1)$ is
	\begin{align*} 
		\P(R_t\in \dd r\,|\, \xib_s=\mathbf{x})
		&= \frac{\int^{\infty}_{z=0}\P(\xib_s\in\dd \mathbf{x}, R_t\in\dd r  \,|\, R_1=z)\P(R_1\in \d z)}
						{\P(\xib_s\in\dd \mathbf{x})  } 
		\\	&= \frac{\int^{\infty}_{z=0}\frac{1}{f_n(z)}f_{n(t-s)}(r-\|\mathbf{x}\|)f_{n(1-t)}(z-r)\dd r \, \nu(\dd z)}
								{\Th_{s}(\|\mathbf{x}\|)}
		\\ 	&= \frac{\Th_{t}(r)}{\Th_{s}(\|\mathbf{x}\|)}f_{n(t-s)}(r-\|\mathbf{x}\|)\dd r.
	\end{align*}
\end{proof}

Note that $\P(R_t\in \dd r\,|\, \xib_s)=\P(R_t\in \dd r\,|\, R_s)$ for $t\in(s,1]$.
This is not surprising since $\{R_s\}$ is a GRB, and hence a Markov process with respect to its natural filtration.

When $\nu_{st}$ admits a density, we denote it by $p_{st}(r)=\nu_{st}(\dd r)/\dd r$.
We see from (\ref{eq:nu_st}) that $p_{st}$ exists for $t<1$.
It follows from the definition of $\th_s$ that $p_{s1}$ only exists if $\nu$ admits a density.

\begin{prop} \label{prop:incLiouville}
	Fix $s\in[0,1)$.
	Given $\xib_s$, the increment $\xib_t-\xib_s$, $t\in(s,1]$, has an $n$-variate Liouville distribution with generating law
	\begin{equation}  \label{eq:nu_star}
		\nu^*(B)=\nu_{st}(B+R_s), \qquad (B\in\Borel),
	\end{equation}
	and parameter vector $\a=(t-s,\ldots,t-s)^\tp$.
\end{prop}
\begin{proof}	
	Consider the case $t=1$ when $\nu$ admits a density $p$.
	In this case the density $p_{s1}$ exists.
	From (\ref{eq:nu_den_td}) and (\ref{eq:nu_s1}), we have
	\begin{align}
		\P(\xib_1-\xib_s \in\dd \mathbf{y}\,|\, \xib_s)&=\frac{\G(n)\e^{R_s}p(\|\mathbf{y}\|+R_s)}{\Th_s(R_s)(\|\mathbf{y}\|+R_s)^{n-1}}
																											\prod_{i=1}^n \frac{y_i^{-s}}{\G(1-s)} \d y_i \nonumber
		\\ &=\frac{\G(n(1-t))p_{s1}(\|\mathbf{y}\|+R_s)}{\|\mathbf{y}\|^{n(1-t)-1}} \prod_{i=1}^n \frac{y_i^{-s}}{\G(1-s)} \d y_i. \label{eq:tran_den}
	\end{align}
	Comparing (\ref{eq:tran_den}) to (\ref{eq:MLD}) shows it to be the law of Liouville distribution 
	with generating law $p_{s1}(x+R_s)\dd x$ and parameter vector $(1-s,\ldots,1-s)^\tp$, as required.
	
	The case when $t<1$ is similar since the density $p_{st}$ exists.
	
	For the final case where $t=1$ and $\nu$ has no density we only outline the proof since the details are far from illuminating.
	Given $\xib_s$, the law of $\xib_1-\xib_s$ is characterised by (\ref{eq:ASP1}).
	We then need to show that this law is equal to the law of $X\mathbf{D}$, 
	where $X$ is a random variable with law $\nu^*$ given by (\ref{eq:nu_star}), 
	and $\mathbf{D}$ is a Dirichlet random variable, independent of $X$, with parameter vector $(1-s,\ldots,1-s)^{\tp}$.
	This is possible by mixing a Dirichlet density with the random scale parameter $X$.
	
\end{proof}

\subsection{Moments}
In this subsection we fix a time $s\in [0,1)$, and we assume that the first two moments of $\nu$ exist and are finite.
\begin{prop}
	The first- and second-order moments of $\xib_t$, $t\in(s,1]$, are 
	\begin{align*}
		&\textnormal{(a)}
		&& \E\left(\left.\xi^{(i)}_{t}\,\right| \xib_s\right)
		 =\frac{1}{n}\mu_1 +\xi_{s}^{(i)},
		\\
		&\textnormal{(b)}
		&& \var\left(\left.\xi^{(i)}_{t}\,\right| \xib_{s}\right)
		 =\frac{1}{n}\left[\left\{\frac{t-s+1}{n(t-s)+1}\right\}\mu_2-\frac{1}{n}\mu_1^2\right],
		\\
		&\textnormal{(c)}
		&& \cov\left(\left.\xi^{(i)}_{t}, \xi^{(j)}_{t} \,\right| \xib_{s}\right)
		 =\frac{t-s}{n}\left\{\frac{\mu_2}{n(t-s)+1}-\frac{\mu_1^2}{n(t-s)}\right\}, \quad (i\ne j),
	\end{align*}
where
	\begin{align*}
		\mu_1&=\frac{t-s}{1-s}\left\{\E(R_1  \,|\, R_{s})-R_{s}\right\},
		\\ \mu_2&=\frac{(t-s)\{1+n(t-s)\}}{(1-s)\{1+n(1-s)\}}\E((R_1-R_s)^2  \,|\, R_{s}).
	\end{align*}
\end{prop}
\begin{proof}
	Fix $0\leq s < t\leq 1$.
	From Proposition \ref{prop:incLiouville}, given $\xib_s$, 
	the increment $\xib_{t}-\xib_{s}$ has an $n$-dimensional Liouville distribution with generating law $\nu^{*}(A)=\nu_{st}(A+R_s)$, 
	and with parameter vector $(t-s,\ldots,t-s)^{\tp}$.
	Defining
	\begin{align*}
		\mu_1 &= \int_0^{\infty} y \, \nu^*(\dd y)
		= \int_{R_s}^{\infty} y \, \nu_{st}(\dd y)-R_s
		= \E(R_t  \,|\, \xib_{s})-R_s,
	\\\intertext{and}
		\mu_2 &= \int_0^{\infty} y^2 \, \nu^*(\dd y)
		= \int_{R_s}^{\infty} (y-R_s)^2 \, \nu_{st}(\dd y)
		= \E((R_t-R_s)^2 \,|\, \xib_{s}),
	\end{align*}
	the equations (a)-(c) in the statement of the proposition hold from \citep[Theorem 6.3]{FKN1990}.
	
	It remains to simplify the expressions for $\mu_1$ and $\mu_2$.
	For this we use two results about \levy random bridges.
	First, from \citep[Corollary 3.10]{HHM1} we can write
	\begin{align*}
	 	\E(R_t \,|\, R_{s})&=\frac{t-s}{1-s}\E(R_1  \,|\, R_{s})+\frac{1-t}{1-s} R_{s}.
	\end{align*}
	The expression for $\mu_{1}$ then follows directly. 
	Second, given $R_s$, the process $\{R_t-R_s\}_{s\leq t\leq 1}$ is a GRB with generating law $\bar{\nu}(B)=\nu_{s1}(B+R_s)$ and activity parameter $n$ \citep[Section 3.7]{HHM1}.
	Hence, given $R_s$,	
	\begin{equation*}
		\{R_t-R_s\}_{s\leq t\leq 1}\law \{X \g_{t1} \}_{s\leq t\leq 1},
	\end{equation*}
	where $X$ is a random variable with law $\bar{\nu}$, and $\{\g_{t1}\}_{s\leq t\leq 1}$ is a gamma bridge with activity parameter $n$, 
	independent of $X$, satisfying $\g_{s1}=0$ and $\g_{11}=1$.
	Note that $\g_{t1}$, $t\in(s,1)$, is a beta random variable with parameters $\a=n(t-s)$ and $\b=n(1-t)$.
	Thus
	\begin{align*}
		\E\left(\left.(R_t-R_s)^2\,\right| R_s\right)&=\E(\g_{t1}^{2})\E(X^2) 
				\\ &=\frac{(t-s)\{1+n(t-s)\}}{(1-s)\{1+n(1-s)\}} \E\left(\left.(R_1-R_s)^2\,\right| R_s\right).
	\end{align*}
\end{proof}

\subsection{Measure change}
In this section we show that $\{\Theta_t(R_t)^{-1}\}$ is a positive martingale with respect to the filtration $\{\F_t\}$.
Through $\{\Theta_t(R_t)^{-1}\}$, we are able to define a new measure $\Q$ under which the ASP $\{\xib_t\}$ is a vector of independent gamma processes.

Fix $0\leq s<t<1$.
Then we have
\begin{align*}
	\E_{\P}\left(\left.\Th_t(R_t)^{-1} \right| \F_s \right)&=\E_{\P}\left(\left. \Th_t(\|\xib_t\|)^{-1} \right| \xib_s \right)
	\\	&=\Th_s(R_s)^{-1}\prod_{i=1}^n\int_{\xi^{(i)}_s}^{\infty}\frac{(y_i-\xi_s^{(i)})^{(t-s)-1}\exp\left\{-(y_i-\xi^{(i)}_s)\right\}}{\G(t-s)} \dd y_i
	\\  &=\Th_s(R_s)^{-1}.
\end{align*}
Noting that $\Theta_0(x)=1$, we see that $\{\Theta_t(R_t)^{-1}\}_{0\leq t <1}$ is a Radon-Nikodym derivative process.
Hence we can define a probability measure $\Q$ by
\[ \left.\frac{\dd \Q}{\dd \P}\right|_{\F_t}=\Th_t(R_t)^{-1}, \qquad (0\leq t<1).\]

\begin{prop} \label{prop:Q}
	Under $\Q$, $\{\xib_t\}$ is a vector of $n$ independent gamma processes such that
	\begin{equation*}
		\Q(\xib_t \in\dd \mathbf{x})=\prod_{i=1}^n \frac{x_i^{t-1}}{\G(t)} \e^{-x_i} \d x_i,
	\end{equation*}
	for $t\in[0,1)$.
\end{prop}
\begin{proof}
	Since $\{\xib_t\}$ is Markov under $\Q$, it suffices to verify that the transition law of $\{\xib_t\}$ is that of a vector of $n$ independent gamma processes.
	For $0\leq s<t<1$, we have
	\begin{align*}
	\Q(\xib_t\in\dd \mathbf{x} \,|\,\xib_s)&=\E_{\Q}(\1\{\xib_t\in\dd \mathbf{x}\}\,|\,\xib_s)
		\\ &=\E_{\P}\left(\left. \frac{\Th_t(R_t)}{\Th_s(R_s)}\1\{\xib_t\in\dd \mathbf{x}\} \right| \xib_s \right)
		\\ &=\prod_{i=1}^n \frac{(x_i-\xi^{(i)}_s)^{(t-s)-1}\exp\{-(x_i-\xi_s^{(i)})\}}{\G(t-s)}\d x_i.
	\end{align*}
\end{proof}

\subsection{Independent gamma bridges representation}
The increments of an $n$-dimensional ASP are identical in law to a positive random variable 
multiplied by the Hadamard product of an $n$-dimensional Dirichlet random variable and a vector of $n$ independent gamma bridges. 
For notational convenience, in this subsection we denote a gamma bridge defined over $[0,1]$ as $\{\g(t)\}$ (instead of $\{\g_{t1}\}$). 

For vectors $\mathbf{x},\mathbf{y}\in\R^n$, we denote their Hadamard product as $\mathbf{x}\circ \mathbf{y}$.
That is,
\begin{equation*}
	\mathbf{x}\circ \mathbf{y}=(x_1y_1,\ldots,x_ny_n)^{\tp}.
\end{equation*}

\begin{prop}
	Given the value of $\xib_s$, the ASP $\{\xib_t\}$ satisfies the following identity in law:
	\[ \{\xib_t-\xib_s \}_{s\leq t\leq 1} \law \{R^* \, \mathbf{D} \circ \boldsymbol{\g}_{t} \}_{s\leq t\leq 1}, \]
	where
	\begin{enumerate}
		\item 
			$\mathbf{D}\in [0,1]^n$ is a symmetric Dirichlet random variable with parameter vector $(1-s,\ldots,1-s)^\tp$;
		\item 
			$\{\boldsymbol{\g}_{t}\}$ is a vector of $n$ independent gamma bridges, each with activity parameter $m=1$, 
			starting at the value 0 at time $s$, and terminating with unit value at time 1;
		\item
			$R^*>0$ is a random variable with law $\nu^*$ given by
			\[ \nu^*(A)=\nu_{s1}(A+R_s);\]
		\item 
			$R^*$, $\mathbf{D}$, and $\{\boldsymbol{\g}_{t}\}$ are mutually independent.
	\end{enumerate}
\end{prop}

\begin{proof}
	Fix $k_i\geq 1$ and the partition $s=t_0^i<t_1^i<\cdots<t_{k_i}^i=1$, for $i=1,\ldots,n$. 
	Then define the non-overlapping increments $\{\D_{ij}\}$ and the vectors $\boldsymbol{\D}$ and $\boldsymbol{\a}$ in a similar way to Section \ref{subsec:charac}.
	The distribution of $\boldsymbol{\D}$
	characterises the finite-dimensional distributions, and hence the law, of the process $\{\xib_t-\xib_s \}_{s\leq t\leq 1}$.
	Note that $\boldsymbol{\D}$ are non-overlapping increments of the master GRB $\{\G_t\}$.
	Thus, given $\xib_s$, $\boldsymbol{\D}$ has a multivariate Liouville distribution with parameter vector $\boldsymbol{\a}$
	and generating law $\nu^*(A)=\nu_{s1}(A+R_s)$, for $t\in(s,1]$. 
	It follows from \citep[Theorem 6.9]{FKN1990} that
	\begin{equation*}
		(\D_{i1},\ldots,\D_{ik_i})^\tp \law R^* \, D_i \mathbf{Y}_i, \qquad \text{for $i=1,\ldots,n$,}
	\end{equation*}
	where 
	(i)   $R^*$ has law $\nu^*$, 
	(ii)  $\mathbf{D}=(D_1,\ldots,D_n)^\tp$ has a Dirichlet distribution with parameter vector $(1-s,\ldots,1-s)^\tp$, 
	(iii) $\mathbf{Y}_i\in[0,1]^{k_i}$ has a Dirichlet distribution with parameter vector $(t^i_1-t^i_0,\ldots,t^i_{k_i}-t^i_{k_i-1})^\tp$, 
	(iv)	$\mathbf{Y}_1,\ldots,\mathbf{Y}_n$, $R^*$, and $\mathbf{D}$ are mutually independent.
	
	Let $\{\g(t)\}_{s\leq t \leq 1}$ be a gamma bridge with activity parameter $m=1$ such that $\g(s)=0$ and $\g(1)=1$. 
	Then the increment vector
	\begin{equation}
		\label{eq:gambr_inc}
		(\g(t^i_1)-\g(t^i_0), \ldots, \g(t^i_{k_i})-\g(t^i_{k_i-1}))^\tp
	\end{equation}
	has a Dirichlet distribution with parameter vector $(t^i_1-t^i_0,\ldots,t^i_{k_i}-t^i_{k_i-1})^\tp$. 
	Hence the increment vector (\ref{eq:gambr_inc}) is identical in law to $\mathbf{Y}_i$. 
	From the Kolmogorov extension theorem, this identity characterises the law of $\{\g(t)\}$.	
	It follows that	
	\begin{equation*}
		\{\xi_t^{(i)}-\xi_s^{(i)}\}_{s\leq t\leq 1} \law \{R^* \, D_i \g_{t}\}_{s\leq t\leq 1}, \qquad \text{for $i=1,\ldots,n$},
	\end{equation*}
	which completes the proof.
\end{proof}

\subsection{Uniform process}
We construct a multivariate process from the ASP $\{\xib_t\}$ such that each one-dimensional marginal is uniformly distributed for each $t\in (0,1]$.

Fix a time $t\in(0,1]$.
Each $\xi^{(i)}_{t}$ is a scale-mixed beta random variable with survival function
\begin{equation*}
	\bar{F}_{t}(x)= \int_x^{\infty} I_{1-x/y}(n-t,t) \, \nu(\dd y),
\end{equation*}
where ${I}_{z}(\alpha,\beta)$ is the regularized incomplete Beta function \citep[6.6]{AS1964}.
The random variables $Y^{(i)}_{t}=\bar{F}_{t}(\xi^{(i)}_{t})$, $i=1,\ldots,n$, are then uniformly distributed.

We now define a process $\{\mathbf{Y}_t\}_{0\leq t\leq 1}$ by
\begin{equation*}
	\mathbf{Y}_t=\left(\bar{F}_t(\xi^{(1)}_t),\ldots, \bar{F}_t(\xi^{(n)}_t)\right)^\tp.
\end{equation*}
By construction, each one-dimensional marginal $Y^{(i)}_t$ is uniform for $t>0$.

\subsection{Application}
We give an example of how an ASP can be used in a copula-interpolation problem.
Consider an $k$-period model where $\mathbf{U}_i$ is a vector of $n$ uniform random variables for $i\in\mathcal{I}=\{1,\ldots,k\}$.
Suppose that $\mathbf{U}_k$ has an Archimedean copula with generator $\psi_{U}$.
Given no further information, how might one simulate the vector $\mathbf{U}_i$, $i<k$, or $(\mathbf{U}_1,\ldots,\mathbf{U}_k)^{\tp}$ in a reasonable way?
The solution we propose is to assume that 
\begin{equation*}
	 {U}_i^{(j)}=\bar{F}_{i/k}\left(\xi_{i/k}^{(j)}\right), \qquad j=1,\ldots,n,
\end{equation*}
where $\{\xib_{t}\}_{0\leq t\leq 1}$ is an $n$-dimensional ASP with the generating law found by substituting $\psi_{U}$ into (\ref{eq:bijection2}).
Simulating a sample path of an ASP is straightforward if one can generate variates from its generating law. 
Simulating $(\mathbf{U}_1,\ldots,\mathbf{U}_k)^{\tp}$ is then a matter of numerically evaluating the survival function $\bar{F}_t$.

In a financial setting, this method could be applied to risk modelling.
Let $\mathbf{X}_i$ be the cumulative log-returns of $n$ assets over the next $i$ days.
Suppose that we wish to simulate $\mathbf{X}_i$, for each $i\in\mathcal{I}$, in order to calculate some risk measure (e.g.~value-at-risk).
Assume that we have estimated the distribution function $G_{ij}$ of each one-dimensional marginal $X^{(j)}_i$
(with sufficient historical data, this is usually straightforward).
Assume further that the Archimedean copula with generator $\psi_U$ provides an adequate fit to historical observations of $(G_{1k}(X^{(1)}_k),\ldots,G_{nk}(X^{(n)}_k))^{\tp}$.
We can then jointly simulate $(\mathbf{X}_1,\ldots,\mathbf{X}_k)^{\tp}$ by simulating $\{\xib_t\}$ and setting $X^{(j)}_i = G_{ij}^{-1}({U}_i^{(j)})$.
Using the ASP in this way imposes a significant amount of structure on the copula of $(\mathbf{X}_1,\ldots,\mathbf{X}_k)^{\tp}$.
Indeed, we have exactly one functional degree of freedom, the choice of $\psi_U$.
This structure may be unnecessarily rigid when data and computational time are abundant.
However, in situations where cross-sectional (and temporal) relationships are uncertain, this structure may provide welcome parsimony; 
and in situations where resources are scarce, reducing the problem of fitting the copula of $(\mathbf{X}_1,\ldots,\mathbf{X}_k)^{\tp}$ to fitting the copula of $\mathbf{X}_k$ may save valuable labour.

In Figure \ref{fig:cop_sim} we show some simulations of $\mathbf{U}_i$ when $n=2$, $k=4$, and 
\[ \psi_U(x)=\int_x^{\infty} (1-x/r) \sqrt{\frac{\l}{2\pi r^3}} \exp\left\{ -\frac{\l(r-\mu)^2}{2\mu^2 r} \right\} \dd r, \]
for $x>0$ and constant $\l>0$ and $\mu>0$.
From (\ref{eq:bijection1}), we see that in this case the generating law of $\{\xib_t\}$ is inverse Gaussian.
In Figure \ref{fig:x_cop_sim} we demonstrate some of the temporal dependencies in $(\mathbf{U}_1,\ldots,\mathbf{U}_4)^{\tp}$.

\begin{figure}[ht]
	\begin{center}
		\subfigure[Interpolated copula $\mathbf{U}_1$]{\includegraphics[scale=0.6]{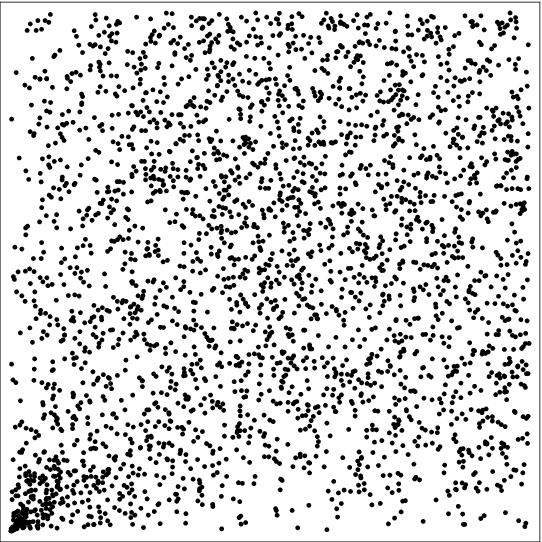}}
		\subfigure[Interpolated copula $\mathbf{U}_2$]{\includegraphics[scale=0.6]{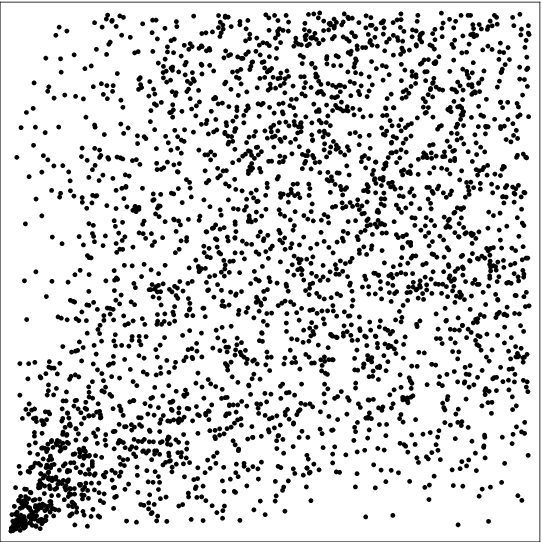}} \\
		\subfigure[Interpolated copula $\mathbf{U}_3$]{\includegraphics[scale=0.6]{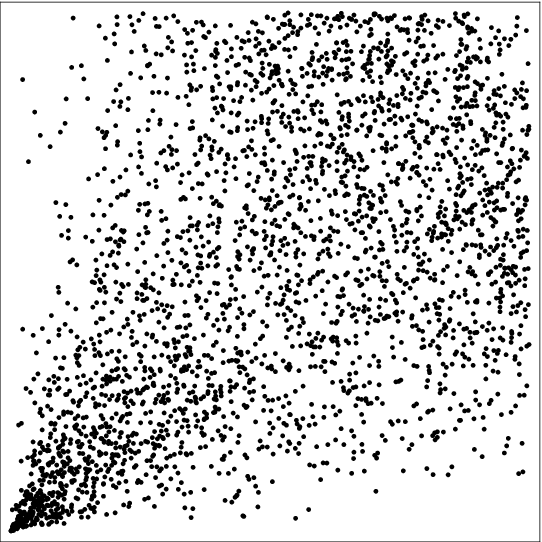}}
		\subfigure[Archimedean copula $\mathbf{U}_4$]{\includegraphics[scale=0.6]{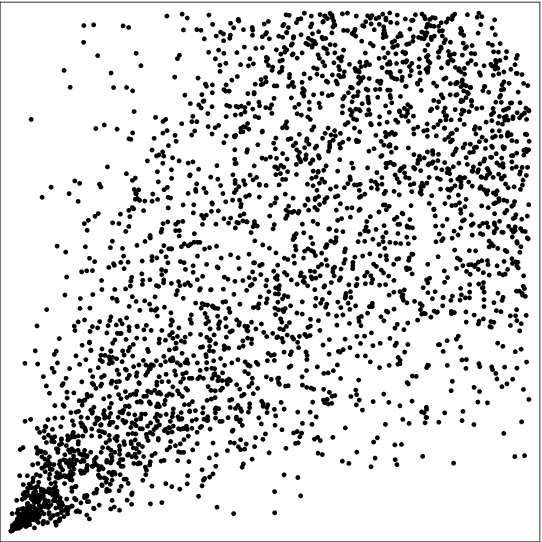}}
	\end{center}
	\caption{Simulations from an Archimedean copula and some interpolated copulas. 
	An inverse-Gaussian law (with parameters $\mu=10$ and $\lambda=0.001$) was used as the generating law of the ASP.
	The interpolation was done at equally-spaced points in the time interval.}
	\label{fig:cop_sim}
\end{figure}

\begin{figure}[ht]
	\begin{center}
		\subfigure[Copula $({U}_1^{(1)},{U}_4^{(1)})$]{\includegraphics[scale=0.6]{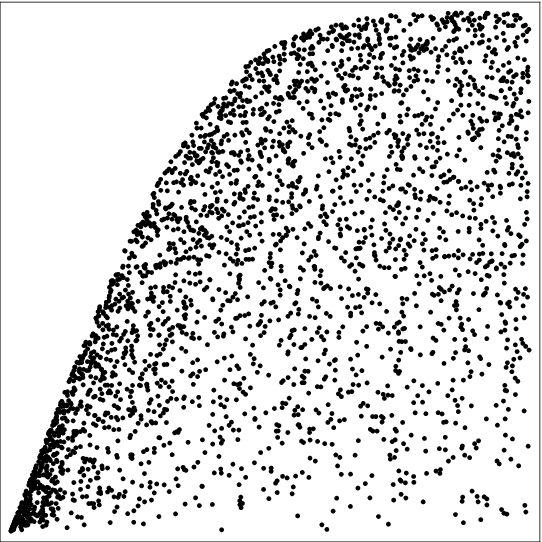}}
		\subfigure[Copula $({U}_2^{(1)},{U}_3^{(1)})$]{\includegraphics[scale=0.6]{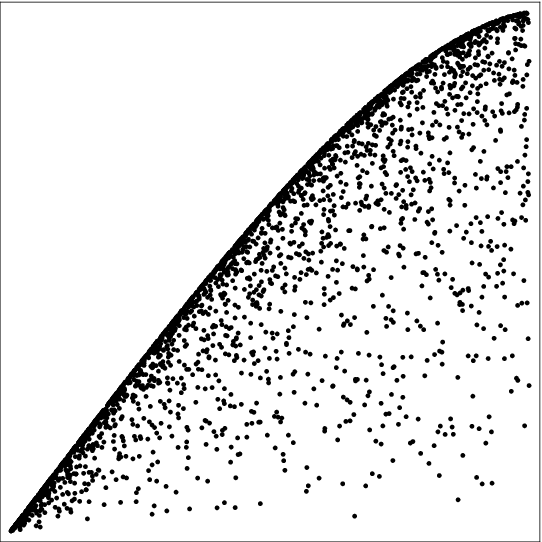}} \\
		\subfigure[Copula $({U}_1^{(1)},{U}_4^{(2)})$]{\includegraphics[scale=0.6]{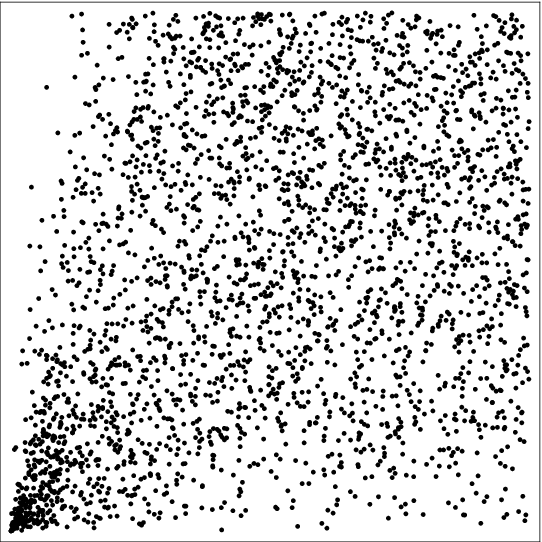}}
		\subfigure[Copula $({U}_2^{(1)},{U}_3^{(2)})$]{\includegraphics[scale=0.6]{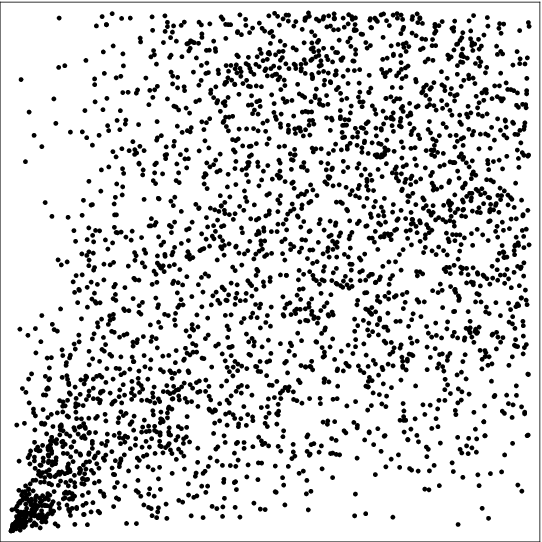}}
	\end{center}
	\caption{Some bivariate-marginal simulations from the eight-dimensional vector $(\mathbf{U}_1,\ldots,\mathbf{U}_4)^{\tp}$.
		Note that $U^{(k)}_j\leq \bar{F}_{j/4}(\bar{F}_{i/4}^{-1}(U^{(k)}_i))$, for $k\in\{1,2\}$ and $1\leq i<j\leq 4$.
		This follows from the monotonicity of one-dimensional marginals of ASPs.}
	\label{fig:x_cop_sim}
\end{figure}

\section{Liouville process} \label{sec:LP}

We generalise ASPs to a family of stochastic processes that we call \textit{Liouville processes}. 
A Liouville process is a Markov process whose increments have multivariate Liouville distributions. 
Liouville processes display a broader range of dynamics than ASPs. 
The one-dimensional marginal processes of a Liouville process are in general not identical.
This generalisation comes at the expense of losing the direct connection to Archimedean copulas.
However, in the language of \citet{MN2010}, the terminal value of a Liouville process has a \emph{Liouville copula};
that is, the survival copula of the terminal value is the survival copula of a multivariate Liouville distribution.

We provide the transition law, moments and an independent gamma bridge representation of a Liouville process.
Proofs are omitted since they are similar to the proofs in Section \ref{sec:ASP}.

\begin{defn} \label{def:LP}
	Fix $n\in\N_+$, $n\geq 2$, and the vector $\mathbf{m}\in\R^n$ satisfying $m_i>0$, $i=1,\ldots,n$.
	Define the strictly increasing sequence $\{u_{i}\}^{n}_{i=1}$ by $u_{0}=0$ and $u_{i}=u_{i-1}+m_{i}$, $i=1,\ldots,n$.
	Then a process $\{\xib_{t}\}_{0\leq t \leq 1}$ satisfying
	\begin{equation*}
		\{\xib_t\}_{0\leq t \leq 1}\law\left\{\left( \G_{t(u_{1})}-\G_{0},\ldots,
		\G_{t(u_{n}-u_{n-1})+u_{n-1}}-\G_{u_{n-1}}
		\right)\right\}_{0\leq t \leq 1},
	\end{equation*}
	for $\{\G_t\}_{0\leq t \leq u_{n}}$ a GRB with activity parameter $m=1$, is an \emph{$n$-dimensional Liouville process}. 
	We say that the generating law of $\{\G_t\}$ is the \emph{generating law} of $\{\xib_{t}\}$ and the \emph{activity parameter} of $\{\xib_{t}\}$ is $\mathbf{m}$.
\end{defn}

Note that allowing the activity parameter of the master process to differ from unity in Definition \ref{def:LP} would not broaden the class of processes.
Indeed, changing the activity parameter of the master process would be equivalent to multiplying the vector $\mathbf{m}$ by a scale factor.

We define a family of unnormalised measures, indexed by $t\in[0,1)$ and $x\in\R_+$, by
\begin{align*}
	\th_0(B;x)&=\nu(B),
	\\	\th_t(B;x)&=\int_B \frac{f_{T(1-t)}(z-x)}{f_T(z)} \, \nu(\dd z),
\end{align*}
for $B\in\Borel$ where $T=\|\textbf{m}\|$.
Again we write	$\Th_t(x)=\th_t([0,\infty);x)$ and $R_t=\|\xib_t\|$.
The process $\{R_t\}$ is a GRB with activity parameter $T$.
Given $\xib_{s}$, the law of $R_{1}$ is $\nu_{s1}(\dd r)= \th_{s}(\dd r;R_s)/\Th_{s}(R_s)$,
and law of $R_{t}$ is
\begin{equation*}
	\nu_{st}(\dd r)=\frac{\Th_{t}(r)}{\Th_{s}(\|\mathbf{x}\|)}\frac{(r-\|\mathbf{x}\|)^{T(t-s)-1}\exp\{-(r-\|\mathbf{x}\|)\}}{\G(T(t-s))}\d r,
\end{equation*}
for $t\in(s,1)$.
Then the Liouville process $\{\xib_t\}$ is a Markov process with the transition law given by
\begin{multline*}
	\P\left(\left. \xi_1^{(1)}\in\dd z_1,\ldots, \xi_1^{(n-1)}\in\dd z_{n-1},\xi_1^{(n)}\in B  \,\right| \xib_s=\mathbf{x} \right)=
	\\ \frac{\th_{\t(s)}(B+\sum_{i=1}^{n-1}z_i;x_n+\sum_{i=1}^{n-1}z_i)}{\Th_s(\|\mathbf{x}\|)}
								\prod_{i=1}^{n-1}\frac{(z_i-x_i)^{m_{i}(1-s)-1}\e^{-(z_i-x_i)}}{\G(m_{i}(1-s))}\d z_i,
\end{multline*}
and
\begin{equation*}
	\P\left( \xib_t\in \d\mathbf{y}  \,|\, \xib_s=\mathbf{x} \right)=
				\frac{\Th_{t}(\|\mathbf{y}\|)}{\Th_s(\|\mathbf{x}\|)}
				\prod_{i=1}^{n}\frac{(y_i-x_i)^{m_{i}(t-s)-1}\e^{-(y_i-x_i)}}{\G(m_{i}(t-s))}\d y_i, 
\end{equation*}
where $\t(t)=1-m_n(1-t)/T$, $0\leq s<t<1$, and $B\in\Borel$.

Similar to an ASP, the joint distribution of the increments of a Liouville process are multivariate Liouville.
In particular, given $\xib_s$ and $t\in(s,1]$, the increment $\xib_t-\xib_s$ has a Liouville distribution with the generating law $\nu^*(B)=\nu_{st}(B+R_s)$, $B\in\Borel$, and parameter vector $(t-s)\mathbf{m}$.
From this we find that, for fixed $0\leq s<t\leq 1$, the first- and second-order moments of $\xib_t$, are
\begin{align*}
	 \E\left(\left.\xi^{(i)}_{t}\,\right| \xib_s\right)
	 &=\frac{m_{i}}{T}\mu_1 +\xi_{s}^{(i)},
	\\
	 \var\left(\left.\xi^{(i)}_{t}\,\right| \xib_{s}\right)
	 &=\frac{m_{i}}{T}\left[\left\{\frac{m_{i}(t-s)+1}{T(t-s)+1}\right\}\mu_2-\frac{m_{i}}{T}\mu_1^2\right],
	\\
	 \cov\left(\left.\xi^{(i)}_{t}, \xi^{(j)}_{t} \,\right| \xib_{s}\right)
	 &=\frac{m_{i}m_{j}(t-s)}{T}\left\{\frac{\mu_2}{T(t-s)+1}-\frac{\mu_1^2}{T(t-s)}\right\}, \quad (i\ne j),
\end{align*}
where
\begin{align*}
	\mu_1&=\frac{t-s}{1-s}\{\E(R_1  \,|\, R_{s})-R_{s}\},
	\\ \mu_2&=\frac{(t-s)\{1+T(t-s)\}}{(1-s)\{1+T(1-s)\}}\E\left(\left.(R_1-R_s)^2  \,\right| R_{s}\right).
\end{align*}

The law of the increments of an $n$-dimensional Liouville process can be characterised by a positive random variable 
multiplied by the Hadamard product of an $n$-dimensional Dirichlet random variable and a vector of $n$ independent gamma bridges.
In particular, given the value of $\xib_s$, $\{\xib_t\}$ satisfies the following identity in law:
\begin{equation*} 
	\{\xib_t-\xib_s \}_{s\leq t\leq 1} \law \{R^* \, \mathbf{D} \circ \boldsymbol{\g}_{t} \}_{s\leq t\leq 1},
\end{equation*}
where
(a) $\mathbf{D}\in [0,1]^n$ has a Dirichlet distribution with parameter vector $(1-s)\textbf{m}$;
(b) $\{\boldsymbol{\g}_{t}\}$ is a vector of $n$ independent gamma bridges, such that the $i$th marginal process is a gamma bridge with activity parameter $m_{i}$, starting at the value 0 at time $s$, and terminating with unit value at time 1;
(c) $R^*>0$ is a random variable with law $\nu^*(B)=\nu_{s1}(B+R_{s})$, $B\in\Borel$;
(d) $R^*$, $\mathbf{D}$, and $\{\boldsymbol{\g}_{t}\}$ are mutually independent.

\subsection{Application}
In financial markets, volatility estimates play an important role in both trading and risk management.
Volatility is unobserved and there is a large body of literature covering its measurement and forecasting (see, for example, \citet{ABDL2003}).
One method of circumventing the intangible nature of volatility is to consider \emph{realized volatility}, or equivalently \emph{realized variance} (RV).
For a given time period (usually one trading day) the RV of an asset is defined as
\[ V=\sum_{i=1}^k\left\{\log(P_i)-\log(P_{i-1})\right\}^2, \]
where $P_0,\ldots,P_k$ are the prices of the asset taken at regular intervals throughout the time period (e.g.~every five minutes).
Realized volatility is then defined as $\sqrt{V}$.

In this application, we model the intraday accumulation of two stock-index RVs using a Liouville process.
A day trader may wish to trade DAX and FTSE futures as a pair in the afternoon, based on some price divergence during the morning.
In order to size the trade appropriately and to manage risk, measures of volatilities for the futures may be required.
Currently, it is common for volatilities to be forecast using information up to the previous day's close-of-business.
However, the proposed model can further incorporate the morning's price movements for an updated, and potentially superior, joint forecast of the afternoon's futures volatilities.
The methods outlined here can be adapted to the modelling of other cumulative phenomena, such as insurance claims as described in Section \ref{sec:intro}.

We define a process $\{\boldsymbol{\eta}_t\}$ taking values in $\R^2$ by
$\eta_t^{(j)}=\xi^{(j)}(\t_j(t))$, $j\in\{1,2\}$, where $\{\xib(t)\}$ is a Liouville process with activity parameter $\mathbf{m}$, and the deterministic time-change function $\t_j$ is continuous, increasing, and satisfies $\t_j(0)=0$ and $\t_j(1)=1$.
Thus $\{\boldsymbol{\eta}_t\}$ is a time-changed Liouville process.
The time change is employed to capture any intraday seasonality observed in market volatility.

We fix a trading day and assume that $\{\boldsymbol{\eta}_t\}$ satisfies
\begin{equation*}
	 \eta_t^{(j)}=\sum_{i=1}^{tk}\left\{\log(P_i^{(j)})-\log(P_{i-1}^{(j)})\right\}^2,
\end{equation*}
for $t=1/k,\ldots,k/k$, $j\in\{1,2\}$, where $\{P_i^{(1)}\}$ and $\{P_i^{(2)}\}$ are the intraday prices of FTSE and DAX futures, respectively.
Thus $t=0$ and $t=1$ are the start and the end of the trading day, and $\boldsymbol{\eta}_1=\xib_1$ is a vector of the day's FTSE and DAX RVs.
See Figure \ref{fig:traj} for a sample path of the accumulation of FTSE and DAX RV during a trading day.

\begin{figure}[ht]
	\begin{center}
		\includegraphics[scale=0.25]{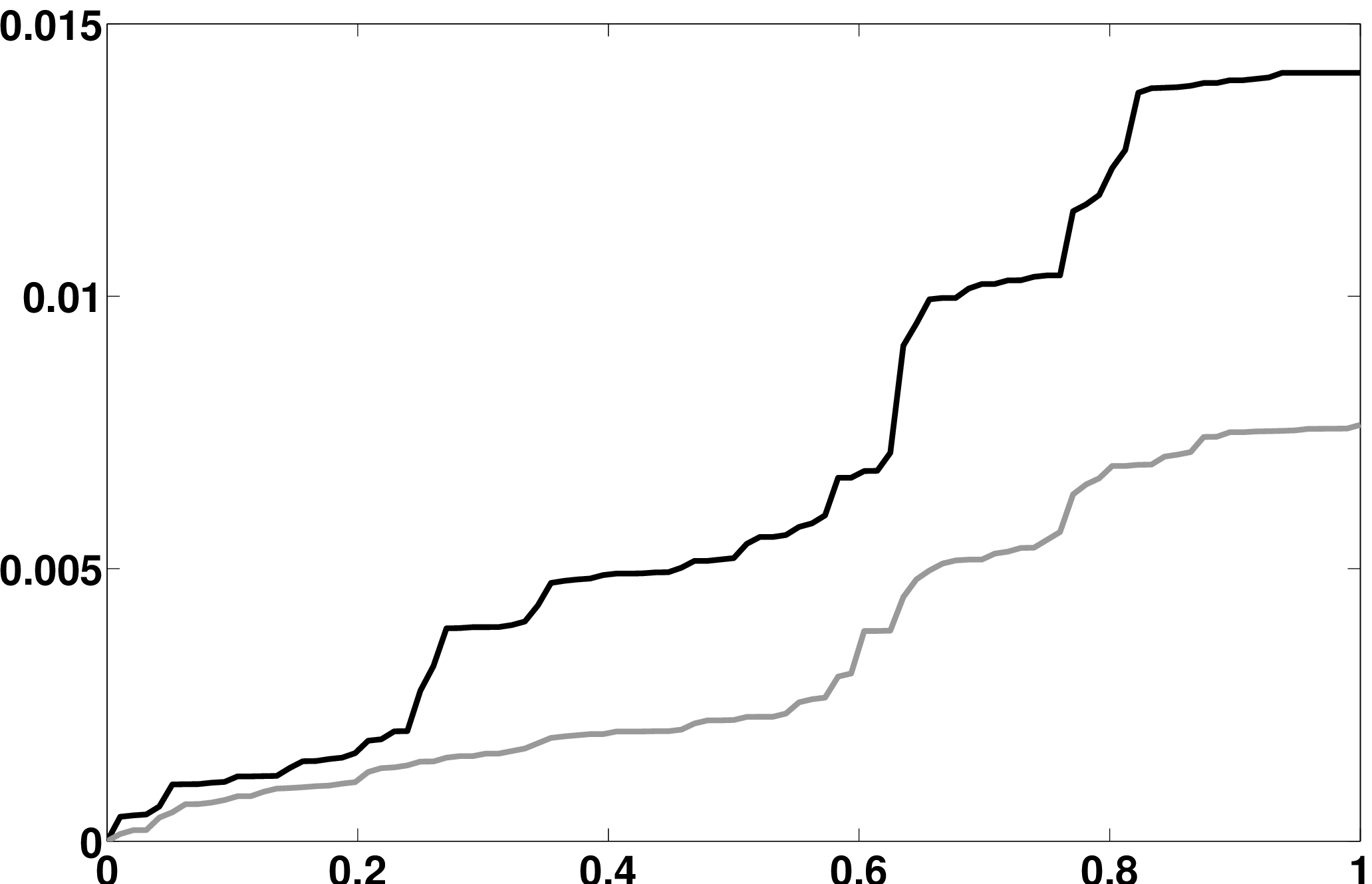}
	\end{center}
	\caption{%
			Sample paths of the accumulation of FTSE RV (black line) and DAX RV (grey line) during a trading day.			}
	\label{fig:traj}
\end{figure}

Before the start of the trading day, we can use a time-series model, fitted to historical data, to provide the generating law of the Liouville process.
The generating law is the law of the sum $R_1$ of the day's FTSE and DAX RV.

We define
\[ X_{ij}=\frac{\eta^{(j)}_{i/k}-\eta^{(j)}_{(i-1)/k}}{R_1}
	=\frac{\left\{\log(P_i^{(j)})-\log(P_{i-1}^{(j)})\right\}^2}{\sum_{i,j}\left\{\log(P_i^{(j)})-\log(P_{i-1}^{(j)})\right\}^2}, \]
then the vector $\mathbf{X}=(X_{1,1},\ldots,X_{k,1},X_{1,2},\ldots,X_{k,2})^\tp$ has a Dirichlet distribution with parameter 
vector $\ab=(m_1 \boldsymbol{\rho}^{(1)},m_2 \boldsymbol{\rho}^{(2)})^{\tp}$, where $\rho_i^{(j)}=\t_j(i/k)-\t_j((i-1)/k)$, for $j\in\{1,2\}$ and $i=1,\ldots,k$.
Once we have the generating law, it remains to estimate the activity parameter $\mathbf{m}$, and the time-change increments $\boldsymbol{\rho}^{(1)}$ and $\boldsymbol{\rho}^{(2)}$.
Using historical intraday prices, these can be jointly fitted using maximum likelihood estimation or moment matching.
This two-stage fitting of the distributions of $R_1$ and $\mathbf{X}$ is natural since they are independent.

To demonstrate the feasibility of this approach, we implemented the model using five-minute prices of FTSE and DAX futures contracts.
We fixed the market opening time ($t=0$) at 9am (GMT) and the market closing time ($t=1$) at 5pm (GMT).
We used the auto-regressive fractionally-integrated model described in \citep{ABDL2003} to construct a log-normal generating law.
The values of $\boldsymbol{\rho}^{(1)}$ and $\boldsymbol{\rho}^{(2)}$ were fitted by moment matching.
An increase in the volatility of FTSE and DAX futures in the European afternoon can be observed in these values coinciding with the opening of financial markets in the US (see Figure \ref{fig:rho}).
Given the values of $\boldsymbol{\rho}^{(1)}$ and $\boldsymbol{\rho}^{(2)}$, we fitted the activity parameter $\mathbf{m}=(53.72,53.15)^{\tp}$ by maximum likelihood estimation.
This value implies that the observed dynamics differ significantly from those of an ASP.
In particular, large values for the activity parameter imply that the observed trajectories increase gradually, exhibiting few large jumps.
This is in contrast to the trajectories of an ASP which increase little between frequent large jumps (cf.~Figure \ref{fig:surprise}).
We used the fitted model to update the joint density of FTSE and DAX RV at various times during a trading day.
As time passes, the joint density converges to a delta function at the actual values of the day's RVs (see Figure \ref{fig:den}).

\begin{figure}[ht]
	\begin{center}
		\subfigure[FTSE $\boldsymbol{\rho}^{(1)}$]{\includegraphics[scale=0.2]{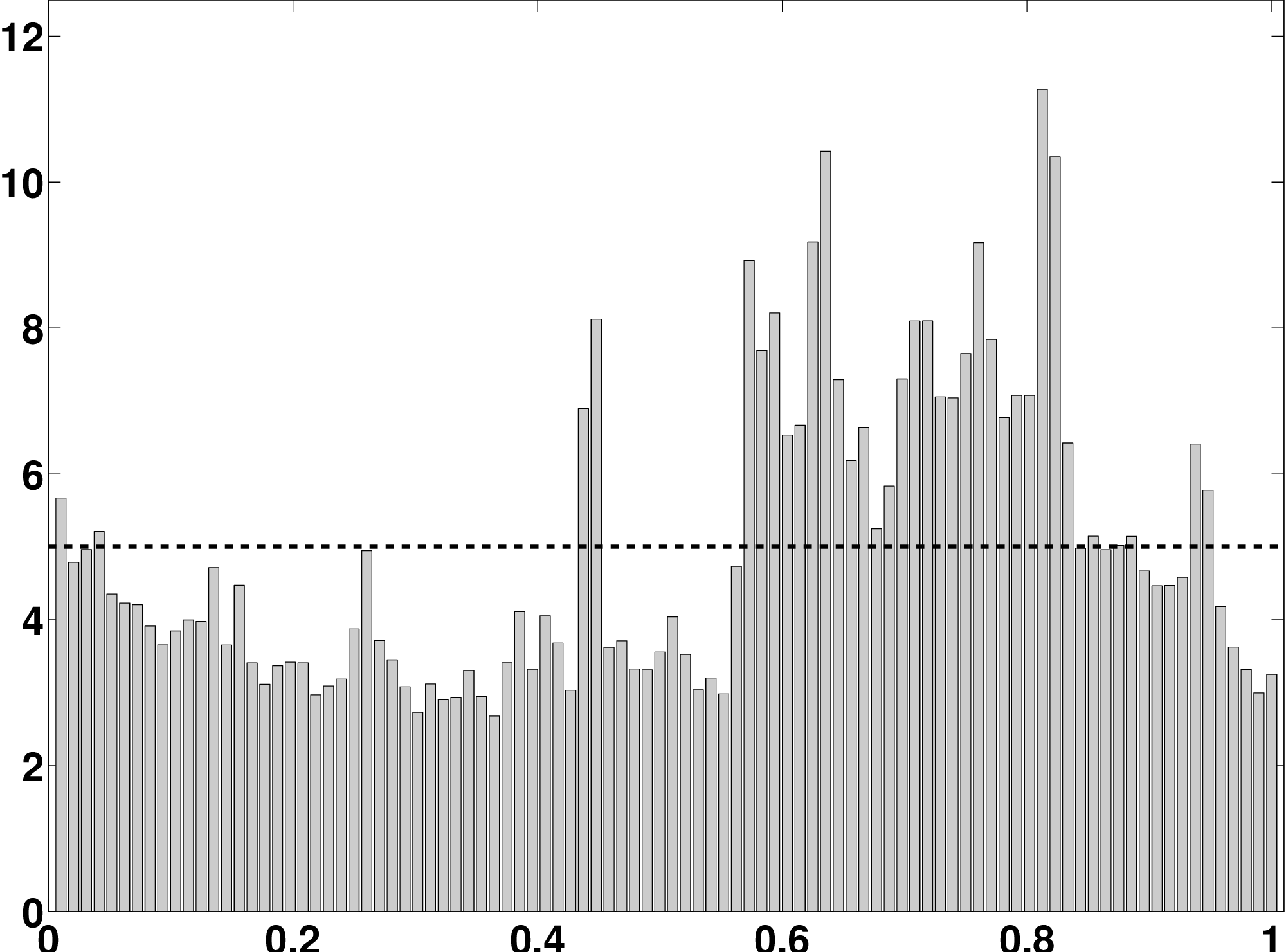}}
		\subfigure[DAX $\boldsymbol{\rho}^{(2)}$]{\includegraphics[scale=0.2]{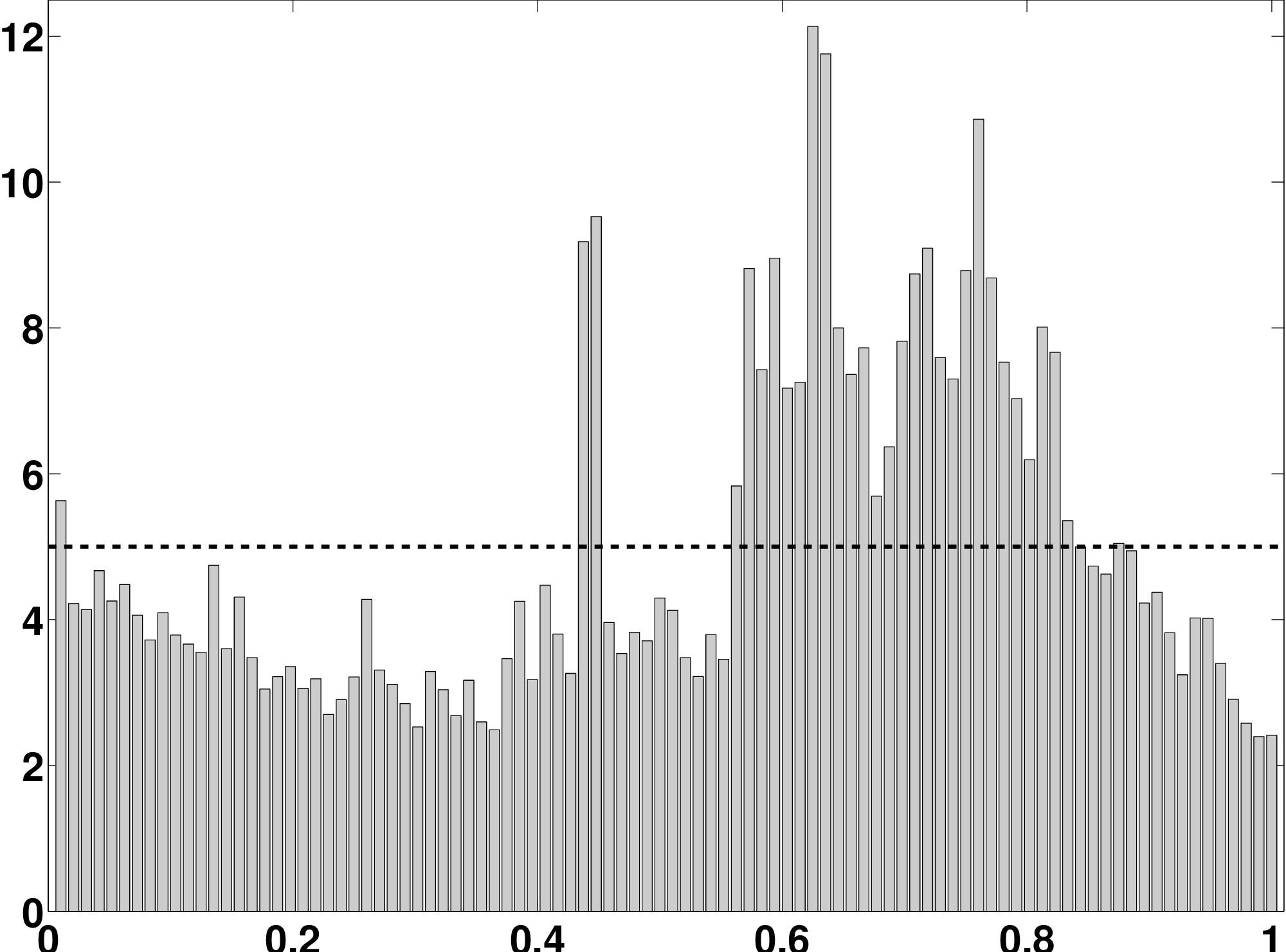}}
	\end{center}
	\caption{%
			Values of $\boldsymbol{\rho}^{(1)}$ and $\boldsymbol{\rho}^{(2)}$ fitted by moment matching.
			The hight of the $i$th bar corresponds to the value of the $i$th element of the vector in units of minutes.
			Each bar represents the number of minutes that the time-change $\t_j$ attributes to each five-minute interval of physical time.
			The higher volatility in afternoon trading manifests as bars with height greater than five.
			}
	\label{fig:rho}
\end{figure}

\begin{figure}[ht]
	\begin{center}
		\subfigure[$t=0$]{\includegraphics[scale=0.25]{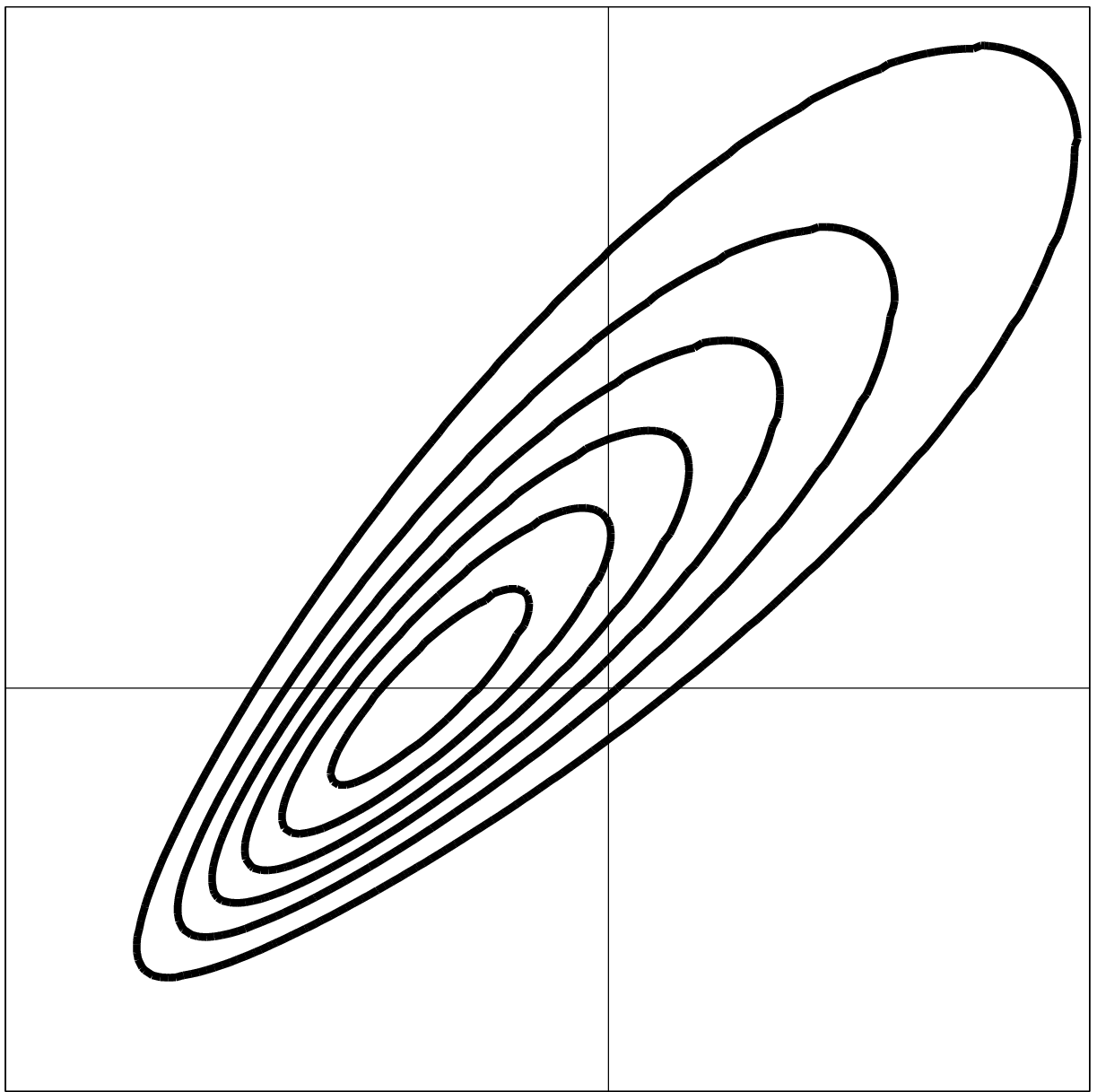}}
		\subfigure[$t=0.25$]{\includegraphics[scale=0.25]{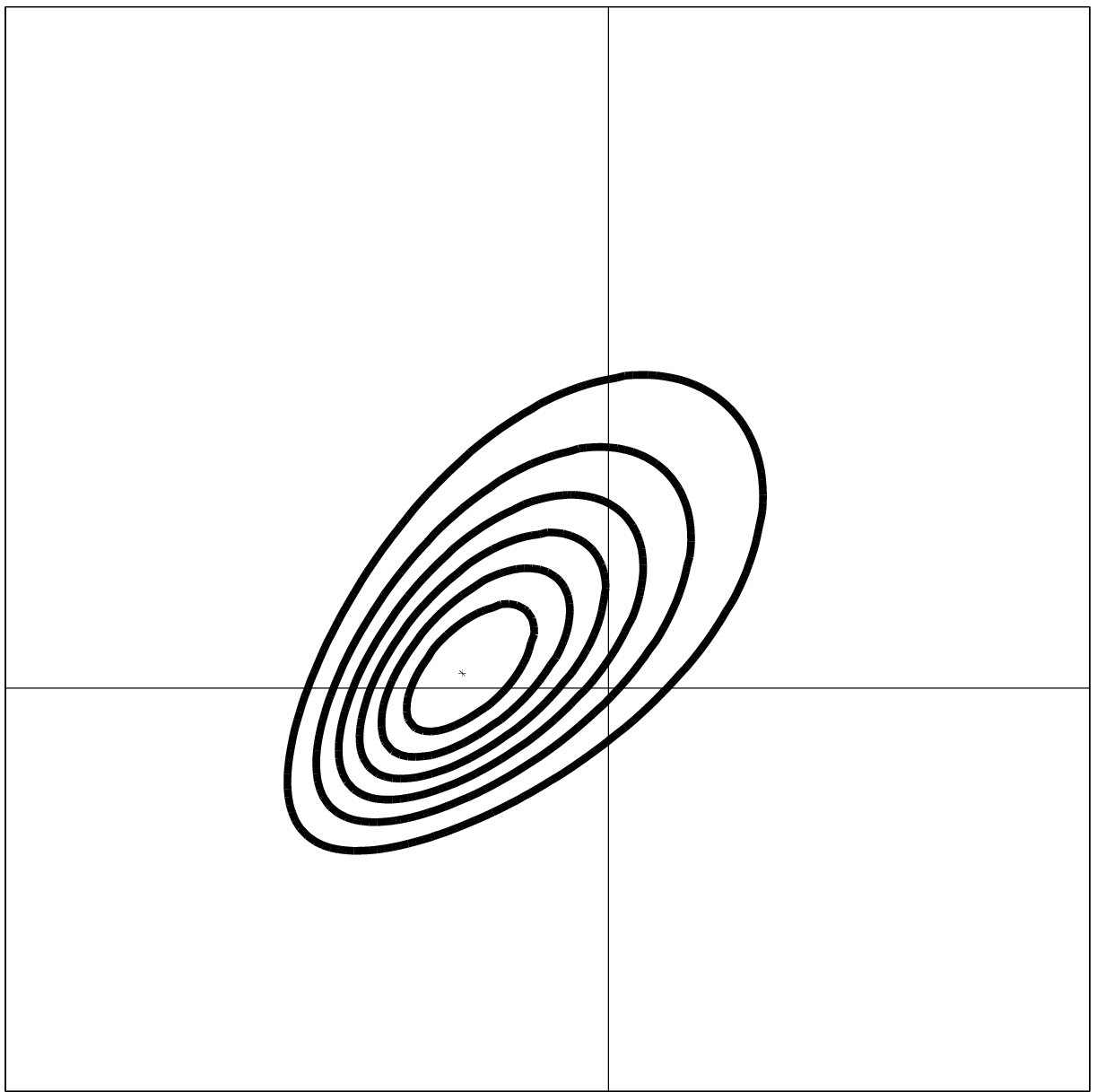}} \\
		\subfigure[$t=0.5$]{\includegraphics[scale=0.25]{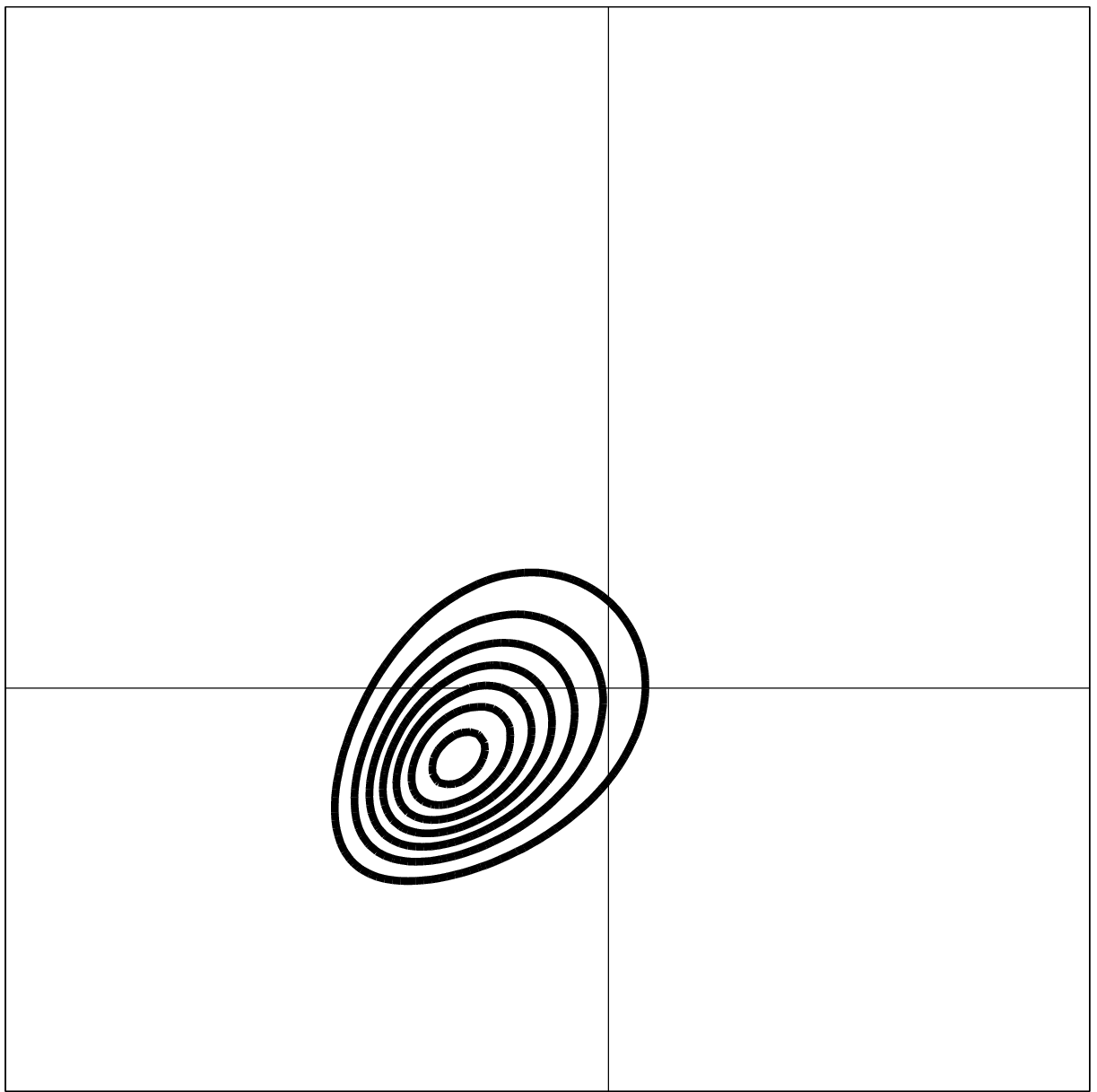}}
		\subfigure[$t=0.75$]{\includegraphics[scale=0.25]{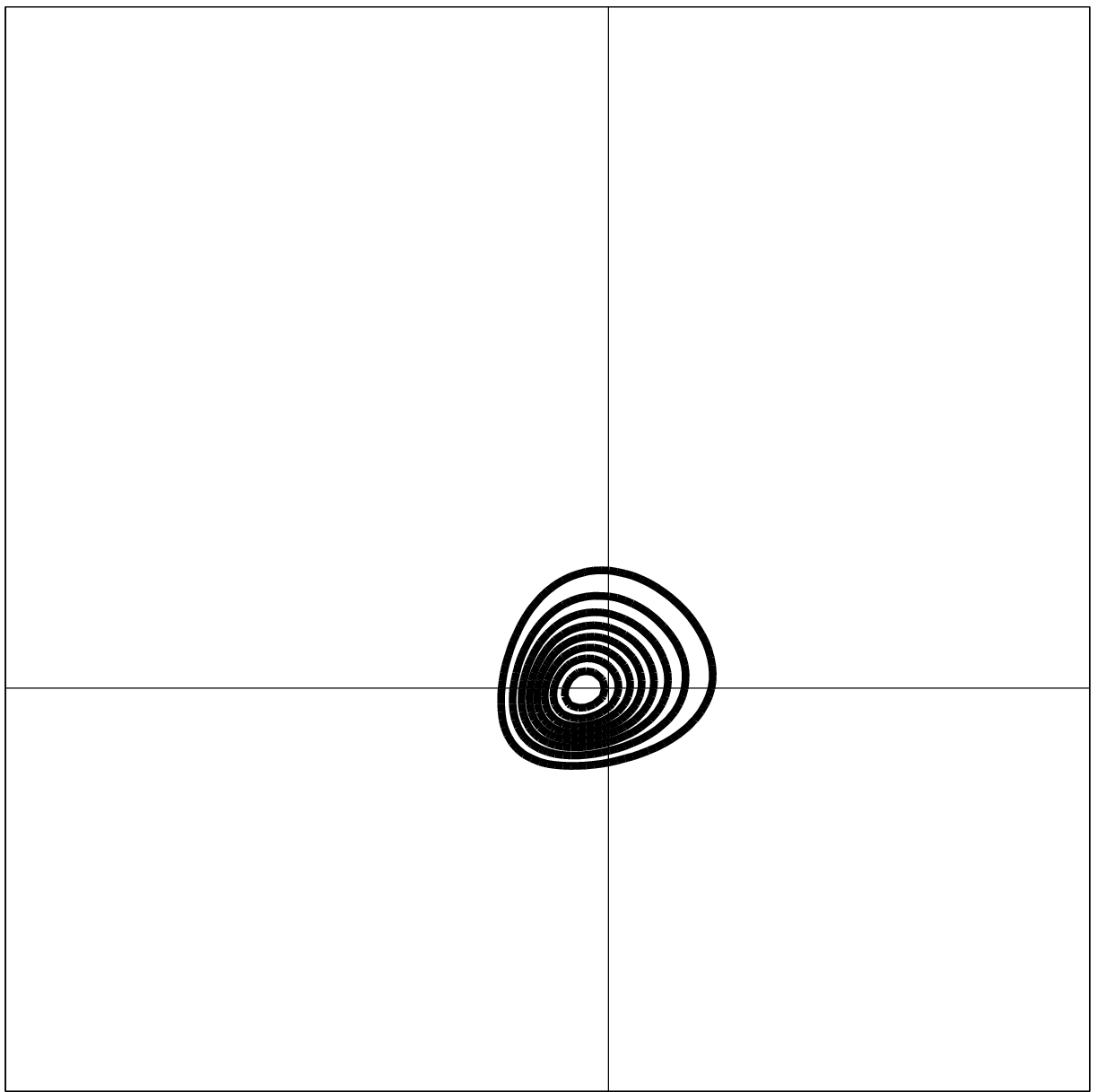}}
	\end{center}
	\caption{%
			Contour plots of the joint density of FTSE and DAX RVs at four times during the trading day.
			The cross hairs indicate the actual value.
			}
	\label{fig:den}		
\end{figure}

\section{Conclusion} \label{sec:conclude}
Through ASPs, we have presented an avenue to extend the theory and application of Archimedean copulas in multi-period and continuous-time frameworks.  
Liouville processes are similarly useful in extending Liouville distributions and Liouville copulas. 
We have also shown that Liouville processes are a natural multivariate extension of GRBs, and thus are a flexible tool in the modelling of cumulative processes.

\section*{Acknowledgements}
The authors are grateful to two anonymous referees whose comments led to significant improvement in this paper.


\bibliography{ASP}

\end{document}